\documentclass[10pt,journal,compsoc]{IEEEtran}
\usepackage[backend=bibtex,style=ieee,natbib=true,mincitenames=1]{biblatex} 
\ifCLASSINFOpdf
\else
\fi
\usepackage{booktabs} % For formal tables
\usepackage{graphicx}   
\usepackage{epstopdf}
\usepackage{xspace}
\usepackage{threeparttable}
\usepackage{multirow}
\usepackage{bbm}
\usepackage{amssymb}
\usepackage[nottoc]{tocbibind}
\usepackage{graphicx}
\usepackage{amsmath}
\usepackage{hyperref}
\usepackage{amsthm}
\newtheorem{theorem}{Theorem}
\newtheorem{lemma}{Lemma}
\newtheorem{proposition}{Proposition}
\newtheorem{definition}{Definition}
\usepackage{caption}
\usepackage[linesnumbered,boxed,ruled,vlined]{algorithm2e}
\usepackage{xcolor}
\newcommand{\eat}[1]{}
\newcommand{\heng}[1]{{{\textcolor{black}{#1}}}} 
\begin{document}
\title{Semi-Supervised Hierarchical Graph Classification}

\author{Jia~Li,
        Yongfeng~Huang,
        Heng~Chang
        and~Yu~Rong%
        %,~\IEEEmembership{Member,~IEEE}% <-this % stops a space
\IEEEcompsocitemizethanks{\IEEEcompsocthanksitem Jia Li and Yongfeng Huang are with Hong Kong University of Science and Technology (Guangzhou) and Hong Kong University of Science and Technology.\protect\\
E-mail: jialee@ust.hk
\IEEEcompsocthanksitem Heng Chang is with Tsinghua University.
\IEEEcompsocthanksitem Yu Rong is the corresponding author and with Tencent AI Lab. \\E-mail: yu.rong@hotmail.com\protect\\
}% <-this % stops an unwanted space
% \thanks{Manuscript received April 19, 2005; revised August 26, 2015.}
}

\IEEEtitleabstractindextext{%
\begin{abstract}
Node classification and graph classification are two graph learning problems that predict the class label of a node and the class label of a graph respectively.  A node of a graph usually represents a real-world entity, e.g., a user in a social network, or a document in a document citation network.  In this work, we consider a more challenging but practically useful setting, in which a node itself is a graph instance.  This leads to a hierarchical graph perspective which arises in many domains such as social network, biological network and document collection.  We study the node classification problem in the hierarchical graph where a ``node'' is a graph instance.  As labels are usually limited, we design a novel semi-supervised solution named SEAL-CI.  SEAL-CI adopts an iterative framework that takes turns to update two modules, one working at the graph instance level and the other at the hierarchical graph level. To enforce a consistency among different levels of hierarchical graph, we propose the Hierarchical Graph Mutual Information (HGMI) and further present a way to compute HGMI with theoretical guarantee. We demonstrate the effectiveness of this hierarchical graph modeling and the proposed SEAL-CI method on text and social network data.  
\end{abstract}

\begin{IEEEkeywords}
graph representation, graph mutual information, hierarchical graph, semi-supervised learning.
\end{IEEEkeywords}}

\maketitle

\IEEEdisplaynontitleabstractindextext

% \fi
\IEEEpeerreviewmaketitle

\IEEEraisesectionheading{\section{Introduction}\label{sec:introduction}}

\IEEEPARstart{G}{raph} has been widely used to model real-world entities and the relationship among them.  Two graph learning problems have received a lot of attention recently, i.e., node classification and graph classification.  Node classification is to predict the class label of nodes in a graph, for which many studies \cite{chang2021spectral,rong2019dropedge,tang2022rethinking} in the literature make use of the connections between nodes to boost the classification performance.  For example, \cite{ramanath2018towards} enhances the recommendation precision in LinkedIn by taking advantage of the interaction network, and \cite{sen2008collective} improves the performance of document classification by exploiting the citation network.  Graph classification, on the other hand, is to predict the class label of graphs, for which various graph kernels \cite{borgwardt2005shortest,gartner2003graph,shervashidze2009efficient,shervashidze2011weisfeiler} and deep learning approaches \cite{Niepert2016LearningCN,DBLP:journals/corr/NarayananCVCLJ17,li2020graph} have been designed.  In this work, we consider a more challenging but practically useful setting, in which a node itself is a graph instance.  This leads to \emph{a hierarchical graph in which a set of graph instances are interconnected via edges}.  This is a very expressive data representation, as it considers the relationship between graph instances, rather than treating them independently.  The hierarchical graph model applies to many real-world data, for example, a social network can be modeled as a hierarchical graph, in which a user group is represented by a graph instance and treated as a node in the hierarchical graph, and then a number of user groups are interconnected via interactions or common members.  As another example, a document collection can be modeled as a hierarchical graph, in which a document is regarded as a graph-of-words \cite{rousseau2015text}, and then a set of documents are interconnected via the citation relationship.  In this paper, we study \emph{hierarchical graph classification, which predicts the class label of graph instances in a hierarchical graph}.

To represent a hierarchical graph, there is a natural and important question:``given multiple levels of inputs and representations of a hierarchical graph, how can we enforce a consistency among different levels of the graph?" In this work, we propose to use \emph{mutual information} (MI) to enforce this consistency. We are motivated by recent developments of graph MI maximization methods \cite{peng2020graph}\cite{sun2019infograph}, and generalize the MI computation to hierarchical graphs, which is named Hierarchical Graph Mutual Information (HGMI). Our theoretical derivations show that HGMI can be decomposed into a linear combination of node-level MI \cite{peng2020graph} and graph-level MI \cite{sun2019infograph}. In this regard, we can use non-hierarchical graph MI computational methods to compute HGMI.  More specifically, for graph instances, we compute MI between nodes and instance representations via graph-level MI computational methods (e.g., INFOGRAPH \cite{sun2019infograph}); for connections between graph instances, we compute MI between instances and hierarchical graph representations via node-level MI computational methods (e.g., GMI \cite{peng2020graph}).

Another challenge is that the amount of available class labels is usually very small in real-world data, which limits the classification performance.  To address this challenge, we take a semi-supervised learning approach to solving the graph classification problem.  We design an iterative algorithm framework which takes turns to update two modules: Instance Classifier (IC) and Hierarchical Classifier (HC).  We start with the limited labeled training set and build IC, which produces the embedding vectors of graph instances.  HC takes the embedding vectors as input and produces predictions.  We cautiously select a subset of predicted labels with high confidence to enlarge the training set.  The enlarged training set is then fed into IC in the next iteration to update its parameters in the hope of generating more accurate embedding vectors.  HC further takes the new embedding vectors for model update and class prediction.  This is our proposed solution, called \underline{SE}mi-supervised gr\underline{A}ph c\underline{L}assification via \underline{C}autious \underline{I}teration (SEAL-CI), to the graph classification problem.

\eat{We also extend this iterative algorithm to the active learning framework, in which we iteratively select the most informative instances for manual annotation, and then update the classifiers with the newly labeled instances in a similar process as described above.  This method is called SEAL-AI in short.}

Our contributions are summarized as follows.

\begin{itemize}
\item We study semi-supervised hierarchical graph classification, which is scarcely studied in the literature.  Our proposed solution SEAL-CI achieves superior classification performance to the state-of-the-other graph kernel and deep learning methods, even when given very few labeled training instances.

\item We generalize the MI estimation to the hierarchical graph domain and propose a new concept named Hierarchical Graph Mutual Information (HGMI). We show HGMI can be decomposed by a sum of mutual information in non-hierarchical graph domain. Upon this, HGMI maximization can be achieved with the help of node/graph-level MI maximization. 

\item We present the HIERARCHICAL GRAPH BENCHMARK \footnote{data and code are available at https://hiergraph.github.io/} with both text data and social network data, with the goal
of facilitating reproducible hierarchical graph research.  From the social networking platform Tencent QQ, we collect 37,836 QQ groups with 18,422,331 unique anonymized users, in which the hierarchical graph is constructed with common memberships (hierarchy-level) and friendships (instance-level). From the arXiv papers, we collect 4,666 papers, in which the hierarchical graph is constructed with citations (hierarchy-level) and semantics (instance-level)
\end{itemize}

The remainder of this paper is organized as follows. Section \ref{sec.related} discusses the related works and Section \ref{def} gives the problem definition. Section \ref{alt} describes the design of SEAL-CI.  We report the experimental results in Section \ref{sec.exp}.  Finally, Section \ref{sec.con} concludes the paper.

\section{Related Work}\label{sec.related}

\textbf{Hierarchical graph representation}
\textcolor{black}{In the literature, most graph representation methods focus on node-level representations \cite{perozzi2014deepwalk}\cite{grover2016node2vec}\cite{kipf2017semi}\cite{defferrard2016convolutional} or graph-level representations \cite{Yanardag:2015:DGK:2783258.2783417}\cite{xu2018powerful}\cite{velivckovic2017graph}. One connection between these two kinds of representations is that graph-level representations could be obtained by integrating node-level representations with graph pooling operations, e.g., DIFFPOOL \cite{ying2018hierarchical}, SAGPool \cite{lee2019self}, Attention-Pool \cite{jiawww19}, MinCut-Pool \cite{bianchi2020mincutpool}. Until recently Li et al. \cite{jiawww19} firstly formulate the hierarchical graph representation problem and evaluate the effectiveness of hierarchical graph representation on social network classification tasks. Wang et al. \cite{wang2020gognn} propose a Graph of Graphs Neural Network (GoGNN)  and evaluate the proposed GoGNN on entity interaction prediction tasks. Xu et al. \cite{xu2021highair} consider city-level graph representations and station-level graph representations to facilitate the air quality forecasting tasks. MIRACLE \cite{wang2021multi} utilizes the multi-view contrast learning method to derive hierarchical graph representations and validated its effectiveness on Drug-Drug Interactions (DDI). However, currently there are still several problems within the hierarchical graph representation research. Firstly, most previous works are based on heuristics and lack the theoretical derivations of consistent representations among different levels of the hierarchical graph. Secondly, most works mildly evaluate the expressness of hierarchical graph representations on one task, which makes the whole story less convincing. }

\textbf{Graph mutual information maximization}
\textcolor{black}{There has been a surging interest in using MI to derive unsupervised graph representations recently. For node-level representations, DGI
\cite{velivckovic2018deep} maximizes MI between a graph summery representation and node-level representations and shows that maximizing this kind of MI is equivalent to maximizing the one between the node features and node-level representations. Later, GMI \cite{peng2020graph} directly approaches MI computation by comparing the node input and node-level representations, without the overheads of an extra graph summery representation. For graph-level representations, INFOGRAPH \cite{sun2019infograph} and MVGRL \cite{icml2020_1971} extend the idea of DGI to the field of graph-level representation and contrasts graph-level and subgraph-level representations. Our work advances this research area by proposing hierarchical graph MI and develops an tractable computation method with guarantees.}

\begin{figure}
\begin{center}
\includegraphics [width=0.35\textwidth]{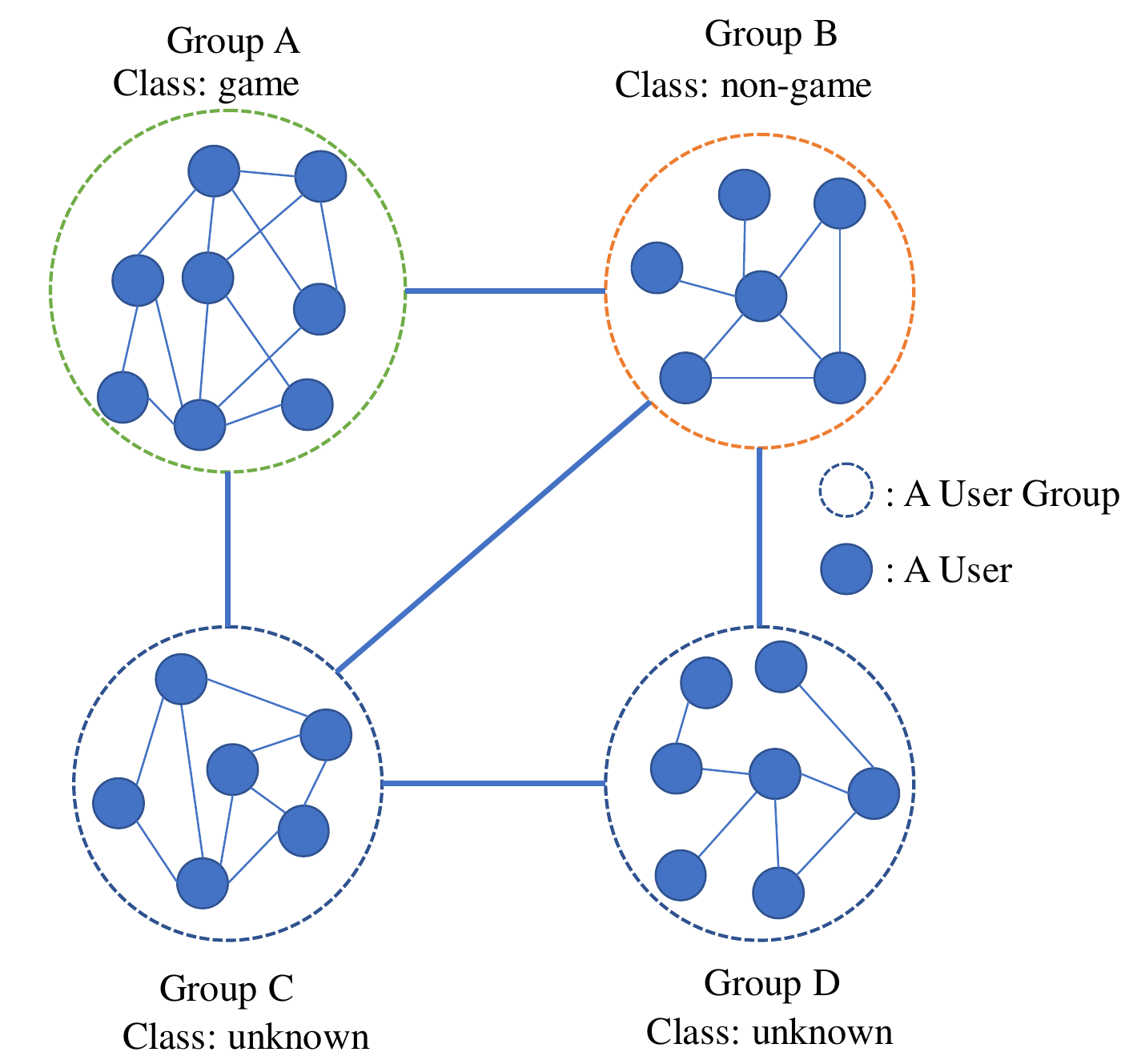}
\end{center}

\caption{\textcolor{black}{A hierarchical graph example with four graph instances $A, B, C, D$, each of which corresponds to a user group in a social network. This figure is adapted from \cite{jiawww19}.}  }
\label{fig.example1}
%\vspace{-0.3cm}
\end{figure}

\textbf{Semi-supervised graph learning}
\textcolor{black}{In many graph applications, the size of labeled instances is typically limited. In addition to labeled instances, semi-supervised graph learning seeks to leverage unlabeled instances within the network structure to obtain gains in performance. In the literature, semi-supervised graph learning techniques can be generally classified into two groups. The first group is consistency regularization. Representative works include graph Laplacian regularization~\cite{zhu2003semi}\cite{yang2016revisiting}, graph contrastive regularization \cite{cg_aaai21}\cite{sun2019infograph}\cite{sun2020multi}. This group aims to learn a representation function that gives consistent predictions for similar data points. The second group is pseudo-label. As the name indicates, this group uses a trained model on the labeled set to produce additional training examples by labeling instances of the unlabeled set \cite{ouali2020overview}. Representative works include ICA \cite{sen2008collective}, Cautious Iteration \cite{mcdowell2007cautious}\cite{cascante2020curriculum}.  In ICA, for each node a local classifier takes the estimated labels of its neighborhood and its own features as input, and outputs a new pseudo label.  The iteration continues until adjacent estimations stabilize. In this work,  our proposed method SEAL-CI connects the idea of  both consistency regularization and pseudo-label. We show our algorithm has a convergence property w.r.t. the empirical risk.}

\eat{Active learning has been integrated in many collective classification methods \cite{settles2012active}\cite{bilgic2010active} to find the most informative samples to be labeled.  However, research that generalizes active learning with deep semi-supervised learning is still sparse.  The closest work is \cite{zhou2017fine} in which the authors utilize active learning to incrementally fine-tune a CNN network for image classification.  Our solution SEAL-AI is different in the sense that the informative samples selected by active learning are used to update the parameters of the graph embedding network, whose output is then fed into HC in an iterative framework.}

\section{Problem Definition}\label{def}
We denote a set of objects as $O=\{o_1, o_2, \ldots, o_N\}$ which represent real-world entities.  We use $d$ attributes to describe properties of objects, e.g., age, gender.

We use a \emph{graph instance} to model the relationship between objects in $O$, which is denoted as $g=(V, A, X)$, $V\subseteq O$ is the node set and $|V|=n$, $A$ is an $n\times n$ adjacency matrix representing the connectivity in $g$, and $X \in \mathbb{R}^{n \times d}$ is a matrix recording the attribute values of all nodes in $g$.

A set of graph instances $G=\{g_1,g_2,\ldots\}$ can be interconnected, and the connectivity between the graph instances is represented by an adjacency matrix $\mathcal{A}$.  The graph instances and their connections are modeled as a \emph{hierarchical graph} $\mathcal{G}=(G,\mathcal{A})$.

A graph instance $g\in G$ is a \emph{labeled graph} if it has a class label, represented by a vector $y\in \{0, 1\}^c$, where $c$ is the number of classes.  A graph instance is \emph{unlabeled} if its class label is unknown.  Then $G$ can be divided into two subsets: labeled graphs $G_l$ and unlabeled graphs $G_u$, where $G=G_l\cup G_u$, $|G_l|=L$ and $|G_u|=U$, $Y = \{y_i\}_{i=1}^L$ denotes the labeled ground truth.  In this paper, we study the problem of \emph{\textbf{graph classification}}, which determines the class label of the unlabeled graph instances in $G_u$ from the available class labels in $G_l$ and the hierarchical graph topological structure.  As the amount of available class labels is usually very limited in real-world data, we take a semi-supervised learning approach to solving this problem.

Our problem formulation is applicable to many real applications as illustrated below.

\textbf{Application 1} Consider a social network where we want to infer the topics of user groups, as depicted in Figure \ref{fig.example1}. $A, B, C, D$ denote four user groups.  Group $A$ has the class label of ``game'', $B$ has the label of ``non-game'', while the class labels of $C$ and $D$ are unknown.  These four groups are connected via some kind of relationships, e.g., interactions or common members.  The internal structure of each user group shows the connections between individual users.  From this hierarchical graph, we want to determine the class labels of groups $C$ and $D$.

\textbf{Application 2} Consider an information retrieval scenario where we want to category documents. We can model a document as a graph-of-words, i.e., graph instances in a hierarchical graph. We then connect these documents by their citation relationships, i.e., connections between graph instances in a hierarchical graph. With this setting, we can achieve better classification performance by taking advantage of both out-of-document relations and inside-of-document semantics (See details in Section \ref{textdata}).

\textbf{Application 3} Protein-Protein Interaction (PPI) sheds light on biological processes, of which one important task is to predict properties of proteins. We can model a protein as a graph of amino-acid interactions \cite{gligorijevic2021structure}, i.e., graph instances in a hierarchical graph. We then connect these proteins by their interactions, forming a hierarchical graph. With this formulation, we can achieve better prediction performance by considering both out-of-protein interactions and inside-of-protein interactions.

%From the above application example, one can expect our method can be applied to social network analysis, information retrieval and Protein property prediction, etc. 
%As another perspective, we can understand our formulation from ``internal and external factors'', in which graph instances models the relations between internal factors and connections between these instances captures the relations among external factors. 

\section{Methodology}\label{alt}

\subsection{Problem Formulation}
In our problem setting, we have two kinds of information: graph instances and connections between the graph instances, which provide us with two perspectives to tackle the graph classification problem.  Accordingly, we build two classifiers: a classifier IC constructed for graph instances and a classifier HC constructed for the hierarchical graph.

For both classifiers, one goal is to minimize the supervised loss, which measures the distance between the predicted class probabilities and the true labels. \eat{Another goal is to minimize a disagreement loss, which measures the distance between the predicted class probabilities by IC and HC.  The purpose of this disagreement loss is to enforce a \emph{\textbf{consistency}} between the two classifiers.}Another goal is to maximize a hierarchical graph mutual information, which measures the distance among the input hierarchical graph, IC representations and HC representations.  The purpose of this hierarchical graph mutual information is to enforce a \emph{\textbf{consistency}} among different levels of hierarchical graph.

Formally, we formulate the graph classification problem as an optimization problem:

\begin{equation}
 \arg \min \zeta(G_l) - \xi(\mathcal{G}),
\label{equ.total}
\end{equation}
where $\zeta(G_l)$ is the empirical risk for the labeled graph instances, and $\xi(\mathcal{G})$ is the hierarchical graph mutual information for all graph instances.

Specifically, $\zeta(G_l)$ includes two parts:
\begin{equation}
\zeta(G_l) = \frac{1}{L}\sum_{g_i\in G_l}(\text{CE}(y_i, e_i) + \text{CE}(y_i, \gamma_i)),
\label{equ.super}
\end{equation}
where $e_i$ is a vector of predicted class probabilities by IC, and $\gamma_i$ is a vector of predicted class probabilities by HC.  $\text{CE}(\cdot, \cdot)$ is the cross-entropy loss function.

The hierarchical graph mutual information (HGMI) $\xi(\cdot)$ is defined as:
\begin{equation}
\xi(\mathcal{G}) = \text{HGMI}(\mathcal{G}) = I(G;E;\Gamma),
\label{equ.unsuper}
\end{equation}
where $I(\cdot;\cdot;\cdot)$ denotes the mutual information between a set of three variables, $E = \{e_i\}_{i=1}^{L+U}$ denotes the representations derived by IC and $\Gamma = \{\gamma_i\}_{i=1}^{L+U}$ denotes the representations derived by HC.  In the following subsections, we first analyse HGMI and describe the way to compute HGMI;  we then give our design of classifiers IC and HC, and our detailed training algorithms.

\eat{The disagreement loss $\xi(\cdot)$ is defined as:
\begin{equation}
\xi(G_u) = \sum_{g_i\in G_u}D_{KL}(\gamma_i || \psi_i),
\label{equ.unsuper}
\end{equation}
where $D_{KL}(\cdot||\cdot)$ is the Kullback-Leibler divergence, $D_{KL}(P||Q) = \sum_jP_j\log \big(\frac{P_j}{Q_j}\big)$.  In the following subsections, we describe our design of classifiers IC and HC, and our approach to minimizing the supervised loss and the disagreement loss.}

\begin{figure}
\begin{center}
\includegraphics [width=0.48\textwidth]{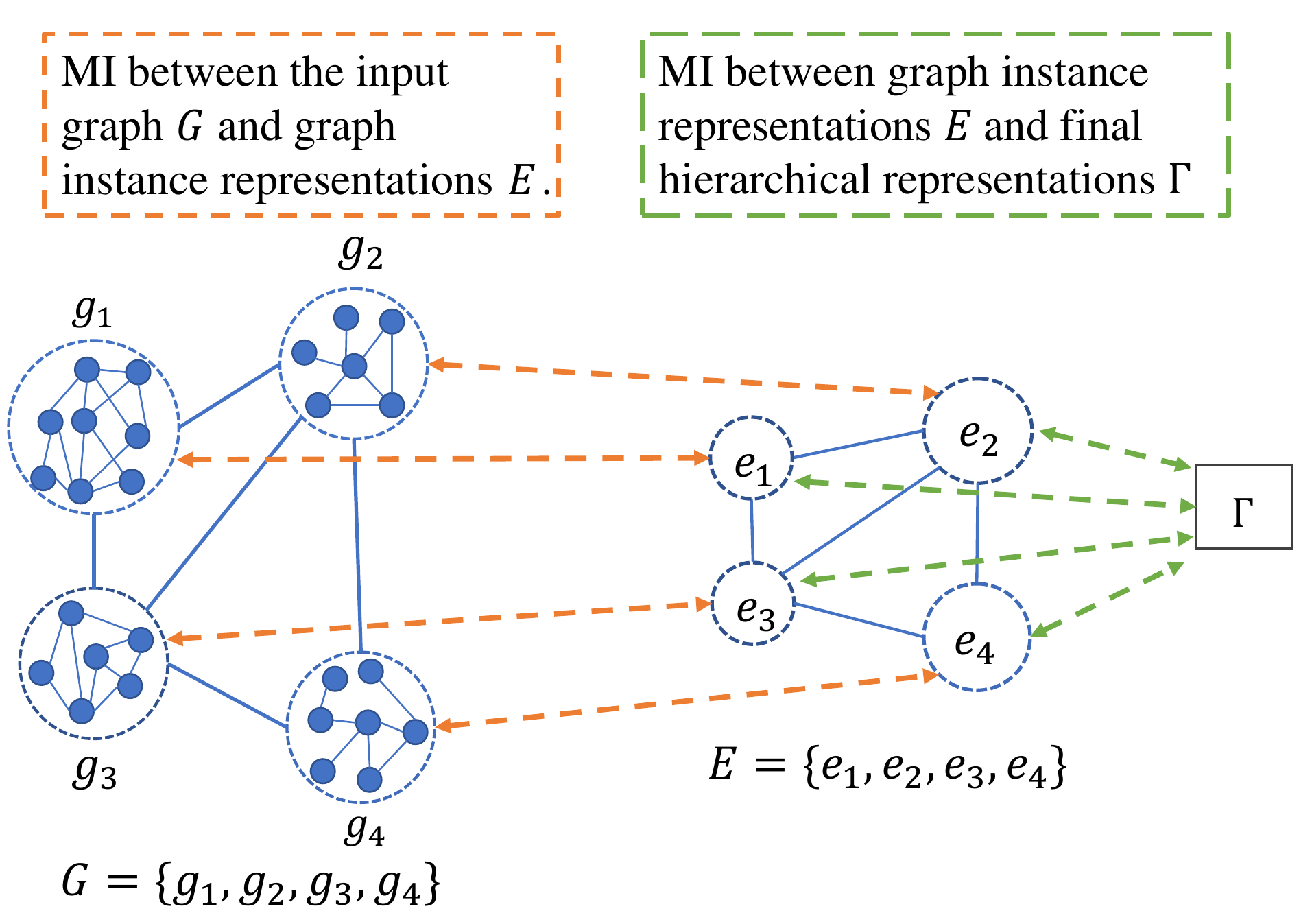}
\end{center}

\caption{Overview of the proposed hierarchical graph mutual information computation. The orange part shows the mutual information computation between the input $G$ and graph instance representations $E$; the green part shows the mutual information computation between graph instance representations $E$ and final hierarchical representations $\Gamma$.}
\label{fig.mi}
%\vspace{-0.3cm}
\end{figure}

\subsection{Analysis of HGMI}\label{hgmi}
We first show HGMI coincides with graph mutual information between $G$ and $\Gamma$ when the hierarchical graph forms a Markov chain.

\begin{theorem}
    \emph{Consider the hierarchical graph forms a Markov chain, i.e., $G \rightarrow E \rightarrow \Gamma$, then
    \begin{equation}
    I(G;E;\Gamma) = I(G;\Gamma),
    \end{equation}
    here $I(G;\Gamma)$ is the mutual information between the input graph instances and hierarchical graph representations.
}
    \label{thm:equa}
\end{theorem}
The proof is trivial, since the three variables form a Markov chain $G \rightarrow E \rightarrow \Gamma$, then $I(G;\Gamma|E) = 0$, we have
    \begin{equation}
    I(G;E;\Gamma) = I(G;\Gamma) - I(G;\Gamma|E) = I(G;\Gamma).
    \end{equation}
It was worth mentioning that the Markov property is quite reasonable here. For example on text data, upon given document representation, it is reasonable to assume graph-of-words inputs (within documents) are irrelevant when predicting document citations (outside of documents).  

We have the following hierarchical graph mutual information decomposition theorem to compute HGMI.
\begin{theorem}
    \emph{Consider the hierarchical graph forms a Markov chain, the hierarchical graph mutual information  can be decomposed by  a sum of non-hierarchical graph mutual information, namely,
    \begin{equation}
    I(G;E;\Gamma) = \alpha(I(G;E) + I(E;\Gamma)),
    \end{equation}
    here $\alpha \in [0,\frac{1}{2}]$, $I(G;E)$ is the mutual information between graph instance input and graph representation on instance level, $I(E;\Gamma)$ is the mutual information between instance representation and hierarchical graph representation on hierarchical level.
}
    \label{thm:submod}
\end{theorem}
 To prove Theorem \ref{thm:submod}, we first introduce two lemmas.
 \begin{lemma}
    \emph{For any random variable $Z_1$, $Z_2$, $Z_3$, we have
    \begin{equation}
    I(Z_1;Z_2,Z_3) \geq \frac{1}{2}(I(Z_1;Z_2) + I(Z_1;Z_3)),
    \end{equation}
    here $I(Z_1;Z_2,Z_3)$ is the mutual information between variable $Z_1$ and the joint distribution of $Z_2$ and $Z_3$ .
}
    \label{thm:l1}
\end{lemma}
 To prove Lemma \ref{thm:l1}, we make use of the chain rule for mutual information.
\begin{align}
\begin{split}
   I(Z_1;Z_2,Z_3) &= I(Z_1;Z_3) + I(Z_1;Z_2|Z_3)\\
   &\geq I(Z_1;Z_3).
\label{geq6}
\end{split}
\end{align}
The last inequality holds as mutual information is non-negative. Accordingly, we have
\begin{equation}
    I(Z_1;Z_2,Z_3) \geq I(Z_1;Z_2).
    \label{geq7}
 \end{equation}
Based on Eq.~\eqref{geq6} and Eq.~\eqref{geq7}, we complete the proof of Lemma \ref{thm:l1}. 
 \begin{lemma}
    \emph{For a Markov chain $G \rightarrow E \rightarrow \Gamma$, we have
    \begin{equation}
    I(E;G,\Gamma) \leq I(E;G) + I(E;\Gamma).
    \end{equation}
}
    \label{thm:l2}
    \vspace{-0.25cm}
\end{lemma}
\emph{Proof}. According to \cite{renner2002mutual}, if the conditional distribution $P(Z_3|Z_1,Z_2)$ is multiplicative, i.e., $\exists$ two functions $s_1$ and $s_2$, s.t., $P(Z_3|Z_1,Z_2) = s_1(Z_3,Z_1)s_2(Z_3,Z_2)$, then $I(Z_1;Z_2) \geq I(Z_1;Z_2|Z_3)$. Since $G \rightarrow E \rightarrow \Gamma$ is a Markov chain, we have
$P(\Gamma|G,E) = P(\Gamma|E)$, which means $P(\Gamma|G,E)$ is multiplicative. Thus, $I(G;E) \geq I(G;E|\Gamma)$ holds in our case.
\begin{align}
\begin{split}
   I(E;G,\Gamma) &= I(E;\Gamma) + I(E;G|\Gamma)\\
   &\leq I(E;G) + I(E;\Gamma).
\label{geq8}
\end{split}
\end{align}
Thus, we complete the proof of Lemma \ref{thm:l2}. 
We then prove Theorem \ref{thm:submod}:
\begin{align}
\begin{split}
   I(G;E;\Gamma) &= I(G;E) - I(G;E|\Gamma)\\
   &= I(G;E) - (I(E;G,\Gamma)-I(E;\Gamma))\\
   &= I(G;E) + I(E;\Gamma) - I(E;G,\Gamma)\\
   &= \alpha(I(G;E) + I(E;\Gamma)).
\label{gce}
\end{split}
\end{align}
The last equality holds as we have $\frac{1}{2}(I(G;E) + I(E;\Gamma)) \leq I(E;G,\Gamma) \leq I(G;E) + I(E;\Gamma)$, based on Lemma \ref{thm:l1} and Lemma \ref{thm:l2}. Thus, Theorem \ref{thm:submod} is proved.

\eat{
\begin{figure*}
\begin{center}
\captionsetup{justification=centering}
\includegraphics [width=0.9\textwidth]{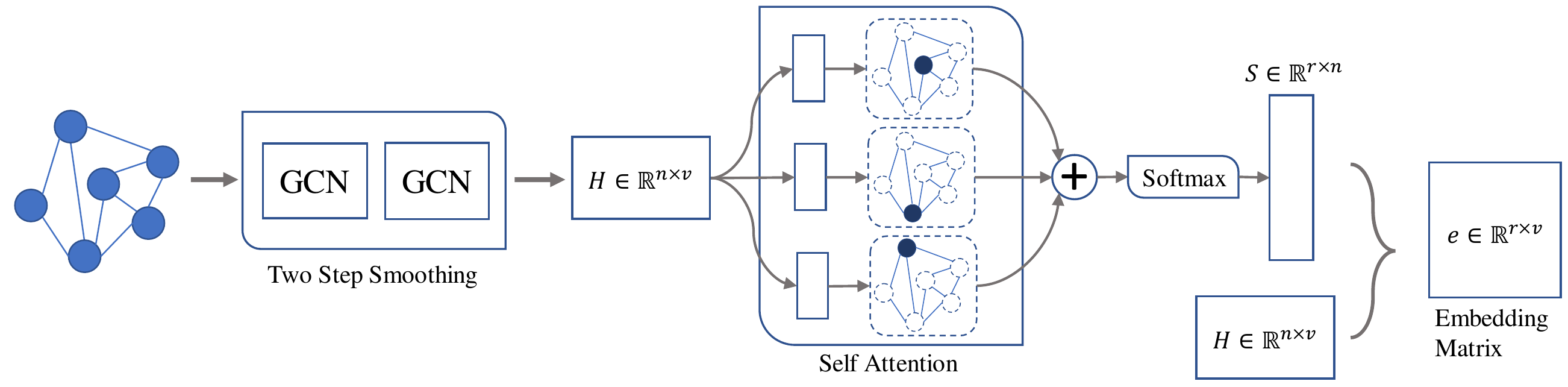}
\caption{The supervised self-attentive graph embedding method SAGE.}
\label{fig.SAGE}
\end{center}
%\vspace{-0.3cm}
\end{figure*}
}

\begin{figure*}
\begin{center}
\includegraphics [width=0.9\linewidth]{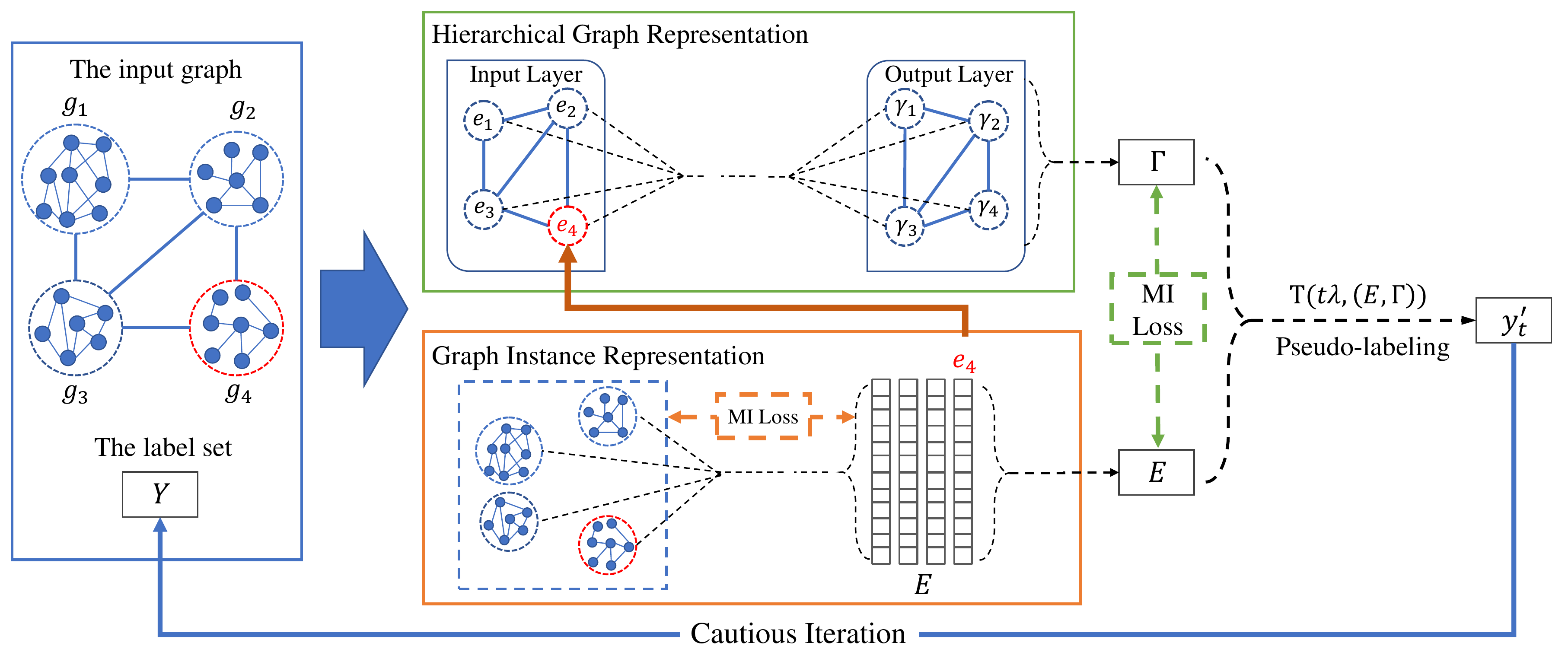}
\caption{\textcolor{black}{Schematic diagram of the learning framework SEAL-CI. There are two subroutines: graph instance representation (in the orange box) and hierarchical graph representation (in the green box).}}
\label{fig.em}
\end{center}
\vspace{-0.3cm}
\end{figure*}

We use Figure \ref{fig.mi} to illustrate the idea to compute HGMI. One naive method to compute HGMI is according to the definition, which inevitably involves with the computation of joint distribution $P(G,E,\Gamma)$ and is not tractable in general \cite{paninski2003estimation}. With Theorem \ref{thm:submod}, we decouple the computation of HGMI into a sum of the computation of graph mutual information without hierarchies, which is tractable as used in many previous works \cite{velivckovic2018deep}\cite{peng2020graph}. Next we present the design of IC and HC classifier, and then give the detailed way to maximize HGMI.

\subsection{Design of Classifiers}\label{demCOCC}
Classifier IC takes a graph instance as input.  As different graph instances have different numbers of nodes, IC is expected to handle graph instances of arbitrary size. Classifier HC takes the hierarchical graph as input, in which individual graph instances are the ``nodes''.  \eat{This is a much too complicated input for a classifier.  To deal with the above challenges,}We propose to embed a graph instance $g_i\in G$ into a fixed-length vector $e_i$ via IC first.  Then HC can take as input the embedding vectors of graph instances and the adjacency matrix $\mathcal{A}$. In particular, IC takes as input the adjacency matrix $A_i$ and attribute matrix $X_i$ of an arbitrary-sized graph instance $g_i$, and outputs an instance representation $e_i$, i.e., $e_i = \text{IC}(A_i, X_i)$.  HC takes the instance representations $E =\{e_i\}_{i=1}^{L+U}$ and $\mathcal{A}$, and outputs the predicted hierarchical graph representations $\Gamma = \{\gamma_i\}_{i=1}^{L+U}$, i.e., $\Gamma = \text{HC}(E, \mathcal{A})$.  In the following, we illustrate the design of IC which performs instance graph representation, and then the design of HC which performs hierarchical graph representation.

\subsubsection{Graph instance representation}\label{dem}

The task of graph instance representation is to produce a fixed-length embedding vector of a graph instance, for which, however, we identify two challenges:

\begin{itemize}
\item  \emph{Size invariance}: How to design the network structure to flexibly take an arbitrary-sized graph instance and produce a fixed-length embedding vector?
\item  \emph{Permutation invariance}: How to derive the representation regardless of the permutation of nodes?
\eat{\item  \emph{Node importance}: How to encode the importance of different nodes into a unified embedding vector?}
\end{itemize}

\eat{In particular, the third challenge is \emph{node importance}, i.e., different nodes in a graph instance have different degrees of importance.  For example, in a ``game'' group the ``core'' members should be more important than the ``border'' members in contributing to the derived embedding vector.  We need to design a mechanism to learn the node importance and then encode it in the embedding vector properly.}

\heng{While the concept of size invariance is straightforward, we specifically provide the definition of permutation invariance on graphs here:}
\begin{definition}[Permutation Invariance on graphs]\label{def.PI}
  Let $P \in \{0,1\}^{n \times n}$ be any permutation matrix of order $n$, then the permutation to a graph instance $g$ is defined as a mapping of the node indices, i.e.:
  \begin{equation*}
    Pg = (PAP^\top, PX).
  \end{equation*}
  Then a function $f$ is invariant to permutation on $g$ if
  \begin{equation*}
    f(A, X) = f(PAP^\top, PX).
  \end{equation*}
\end{definition}

Recently \emph{graph pooling} has emerged as a research topic that can be used to tackle the above challenges. Some common practices include DIFFPOOL \cite{ying2018hierarchical}, Attention-based Pool \cite{jiawww19}\cite{lee2019self}\cite{baek2021accurate}, MinCut-Pool \cite{bianchi2020mincutpool}. In particular, \cite{baek2021accurate} shows Attention-based Pool is permutation-invariant by connecting with 
Weisfeiler-Lehman test. To this end, in IC, we first utilize a multi-layer Graph Neural Network (GNN)~\cite{kipf2017semi} to smooth each node's features over the graph's topology.  Then we adopt a \emph{graph pooling} method  and transform a variable number of smoothed nodes into a fixed-length graph instance representation $e$.

\eat{To this end, we propose a self-attentive graph embedding method, called SAGE, which can take a variable-sized graph instance, and combine each node to produce a fixed-length vector according to their importance within the graph.  In SAGE, we first utilize a multi-layer Graph Neural Network (GNN)~\cite{kipf2017semi} to smooth each node's features over the graph's topology.  Then we use a self-attentive mechanism to learn the node importance and then transform a variable number of smoothed nodes into a fixed-length embedding vector, as proposed in ~\cite{DBLP:journals/corr/LinFSYXZB17}.  Finally, the embedding vector is cascaded with a fully connected layer and a softmax function, in which the label information can be leveraged to discriminatively transform the embedding vector $e$ into $\psi$.  Figure~\ref{fig.SAGE} depicts the overall framework of SAGE.}

Formally, we are given the adjacency matrix $A \in \mathbb{R}^{n \times n}$ and the attribute matrix $X \in \mathbb{R}^{n \times d}$ of a graph instance $g$ as inputs.
We apply a multi-layer GNN network:
\begin{equation}
  H = \text{GNN}(A,X),
\label{equ.H}
\end{equation}
\eat{In the preprocessing step, the adjacency matrix $A$ is normalized:
\begin{equation}
  \hat{A} = \tilde{D}^{-\frac{1}{2}}(A + I_n)\tilde{D}^{-\frac{1}{2}},
\label{equ.gene}
\end{equation}
where $I_n$ is the identity matrix and $\tilde{D}_{ii} = \sum_m\ (A + I_n)_{im}$.  Then we apply a two-layer GCN network:
\begin{equation}
  H = \hat{A}\ \text{ReLU}(\hat{A}XW^0)W^1.
\label{equ.H}
\end{equation}

Here $W^0 \in \mathbb{R}^{\phi \times h}$ and $W^1 \in \mathbb{R}^{h \times v}$ are two weight matrices.}
here we get a set of representation $H \in \mathbb{R}^{n \times v}$ for nodes in $g$.  Note that the representation $H$ is size variant, i.e., its size is still determined by the number of nodes $n$.  So next we utilize the \emph{graph pooling} mechanism:
\begin{equation}
e = \text{Pooling(H)},
\end{equation}
here consider Attention-based Pool is adopted, as the attention weight is used to multiply with $H$ and flatten the representation, $e $ is size invariant and does not depend on the number of nodes $n$ any more, \heng{which solves the first challenge.}

\heng{
While for the second challenge of the node permutations, the following positive Proposition shows that our constructed instance-level classifier IC is permutation invariant as long as the selected component graph pooling method and GNN are permutation invariant:}
\begin{proposition}
Let $P \in \{0,1\}^{n \times n}$ be any permutation matrix, then $IC(A,X) = IC(PAP^\top, PX)$ as long as the chosen graph pooling method and GNN are permutation invariant.
\end{proposition}
\begin{proof}
Consider a permutation invariant GNN, e.g., GIN \cite{xu2018powerful}, we have:
\begin{equation}\label{equ.gnn-PI}
  H = \text{GNN}(A,X) = \text{GNN}(PAP^\top, PX).
\end{equation}
For pooling operation, consider the permutation invariant Attention-based Pooling function used in ~\cite{jiawww19}:
\begin{equation}\label{equ.att-pool}
  S = \textsf{softmax} \big(W_{s2}\textsf{tanh}(W_{s1}H^T)\big), \quad e = SH,
\end{equation}
where $W_{s1}$ and $W_{s2}$ are two weight matrices. $H$ is the same after nodes permutation from Eq.~\eqref{equ.gnn-PI} and our example pooling function Eq.~\eqref{equ.att-pool} is only relevant to $H$, which reveals the permutation invariance of IC. Thus we finish the proof.
\end{proof}
\eat{
% The proof is trivial and we omit is here. 
Since the commonly used graph pooling methods and GNNs are usually permutation invariant, the permutation invariance of IC can be easily guaranteed during practice.
}

To summarize, we use GNN and \emph{graph pooling} that are permutation invariant to construct the instance-level classifier IC.  It produces graph instance representations $E = \{e_i\}_{i=1}^{L+U}$ while successfully dealing with the aforementioned two challenges, which is the input for classifier HC in the next.

\subsubsection{Hierarchical graph representation}\label{ssc}
Given $E$ and the adjacency matrix $\mathcal{A} \in \mathbb{R}^{(L+U) \times (L+U)}$, our next task is to infer the parameters of classifier HC and derive the predicted probabilities $\Gamma = \{\gamma_i\}_{i=1}^{L+U}$. This problem falls into the setting of non-hierarchical graph setting where $E$ can be treated as the set of node features. \eat{ Recently neural network based approaches such as~\cite{kipf2017semi,yang2016revisiting} have demonstrated their superiority to traditional methods such as ICA~\cite{sen2008collective}.  In this context we make use of GCN~\cite{kipf2017semi} again for the consideration of efficiency and effectiveness.} In this context, we consider again a multi-layer GNN. Thus the model becomes:
\begin{equation}
  \Gamma = HC(E,\mathcal{A}) = \text{softmax}(\text{GNN} (E,\mathcal{A})),
\label{Eq:GCN}
\end{equation}
where $\Gamma \in \mathbb{R}^{(L+U) \times c}$ is the derived hierarchical graph representation.  With $\Gamma$ and $E$, next we introduce the way to maximize HGMI.

\subsection{Maximization of HGMI}\label{sec.max}
With Theorem \ref{thm:submod}, the computation of HGMI is reduced into MI between graph input  and graph-level representation ($I(G;E)$) and MI between graph input and node-level representation ($I(E;\Gamma)$), in which we have many choices such as INFOGRAPH \cite{sun2019infograph} and GMI \cite{peng2020graph}. 

To compute $I(G;E)$, we use the method of INFOGRAPH \cite{sun2019infograph}, where MI is computed between the graph input $G$ and graph-level representation $E$. More specifically, it is shown in \cite{sun2019infograph} that maximizing the global $I(G;E)$ can be estimated by MI between node representations and graph-level representations,
\begin{equation}
   I(G;E) = \sum_i^{L+U}\frac{1}{L+U}\sum_j^{n}\frac{1}{n}I(h_{ij};e_i),
   \label{hgmi_i}
\end{equation}
where $h_{ij}$ is a node representation sampled from $H_i$.
The core thus becomes how we compute $I(h_{ij};e_i)$. In this work, we resort to Jensen-Shannnon MI estimator (JSD),
\begin{equation}
  I(h_{ij};e_i) =  - sp(-\mathbb{D_I}(h_{ij},e_i) -  \mathbb{E}_{\hat{h}_{ij}}sp (\mathbb{D_I}(\hat{h}_{ij},e_i)),
\label{equ.Hz}
\end{equation}
where $\mathbb{D_I}: \mathbb{R}^v \times \mathbb{R}^m \rightarrow \mathbb{R}$ is a discriminator constructed by a neural network with parameters $D_I$, i.e., \textcolor{black}{$\mathbb{D_I}(h_{ij},e_i) = \sigma(h_{ij}^{\top}W_{D_I}e_i)$ and $W_{D_I}$ is a learnable scoring matrix}. $\hat{h}_{ij}$ is a negative example. $sp(x) = \log(1+\exp(x))$ is the soft-plus function. For a detailed graph-level negative sampling process, please see INFOGRAPH \cite{sun2019infograph}.

To compute $I(E;\Gamma)$, we adopt the method of GMI \cite{peng2020graph}, where MI is computed between the graph feature input and node-level representation. More specifically, it is shown in  \cite{peng2020graph} that, for a node-level representation $\gamma_i$, maximizing the global $I(E;\gamma_i)$ can be decomposed as a weighted sum of local MIs,
\begin{equation}
   I(E;\gamma_i) = \sum_j^{L+U}w_{ij}I(e_j;\gamma_i),
\label{hgmi_h}
\end{equation}
here we adopt the mean version of GMI, meaning $w_{ij} = \frac{1}{L+U}$. To compute $I(e_j;\gamma_i)$, we use JSD again,
\begin{equation}
  I(e_j;\gamma_i) = - sp(-\mathbb{D_H}(e_j,\gamma_i) -  \mathbb{E}_{\hat{e}_j}sp (\mathbb{D_H}(\hat{e}_j,\gamma_i)),
\label{equ.Hsed}
\end{equation}
where $\mathbb{D_H}: \mathbb{R}^m \times \mathbb{R}^c \rightarrow \mathbb{R}$ is a discriminator constructed by a neural network with parameters $D_H$. $\hat{e}_j$ is a negative example. For a detailed node-level negative sampling process, please see GMI \cite{peng2020graph}.
\subsection{The Proposed SEAL-CI Model}\label{sec.method}
In this subsection, we present our method to minimize the objective function Eq.~\eqref{equ.total}. A naive way would be directly minimizing Eq.~\eqref{equ.total} in an end-to-end fashion with the limited labels. \textcolor{black}{This algorithm is called \underline{SE}mi-supervised gr\underline{A}ph c\underline{L}assification (SEAL).} However, in real-world scenarios, the number of labeled graph instances $L$ can be quite small compared to the number of unlabeled instances $U$.  In this context, neural network based classifiers may suffer from the problem of overfitting.

Following previous works \cite{mcdowell2007cautious}\cite{sen2008collective}, we use the idea of iterative algorithm to alternate optimizing the two modules of IC and HC by trusting a subset of predictions.
To be more specific, we combine the instance representation in Section \ref{dem} and hierarchical representation in Section \ref{ssc} into one iterative algorithm.  We build IC to produce instance representation $E^t$ for all graph instances in iteration $t$, and then feed $E^t$ into HC to get the predicted probabilities $\Gamma^t$.  We then make use of pseudo-labeling strategy to update the parameters of IC and generate $E^{t+1}$, which is then used as the input of HC in iteration $t+1$.  Figure~\ref{fig.em} depicts the overall framework of this iterative process. 

\begin{algorithm}[t]
  \caption{SEAL-CI}
  \label{algo.naive}
  \KwIn{$\mathcal{G}=\{G,\mathcal{A}\}$.}
  \KwOut{$\Gamma^{t+1}$,$E^{t+1}$.}
  Initial: $G_{tmp} =\emptyset$, $G_l^0 = G_l$, $t = 0$;

  \While{$t\lambda \leq U$}{
    $\mathcal{W}^{t+1} \leftarrow$ $\arg\min\zeta(G_l^{t}| \mathcal{W}^{t}) - \xi(\mathcal{G}^t| \mathcal{W}^{t})$\;

    $ E^{t+1} \leftarrow$ \textsf{IC}($A, X | \mathcal{W}^{t+1}$)\;
    $ \Gamma^{t+1} \leftarrow$ \textsf{HC}($E^{t+1}, \mathcal{A} | \mathcal{W}^{t+1}$)\;
    %\While{$|G_{tmp}| < \lambda$ and $G_u^{t} \neq \emptyset$}{
    %    $g_i = \arg\max_{g_i \in G_u^t, \gamma_{g_i} \in \Gamma^{t+1}}\gamma_{g_i}$\;
    %    $G_{tmp} \leftarrow g_i$\;
    %    $G_u^{t} \leftarrow G_u^{t} \setminus g_i$\;
    %}
	$G_{tmp} \leftarrow T(t\lambda,(E_{G_u}^{t+1},\Gamma_{G_u}^{t+1}) )$\;
	%$G_u^{t+1} \leftarrow G_u^{t} \setminus G_{tmp}$\;
    $G_l^{t+1} \leftarrow G_l \cup G_{tmp}$\;
    %$G_u^{t} \leftarrow G_u^{t+1}$\;
	$G_{tmp} =\emptyset$\;
	%$t \leftarrow t + 1$\;
  }
  Return  $\Gamma^{t+1}$,$E^{t+1}$\;
\end{algorithm}

\subsubsection{Pseudo-labeling}
To update the instance representations, a naive approach is feeding the whole set of ($E^t$,$\Gamma^t$) as pseudo labels for the parameter update in IC, which is the idea of the original ICA~\cite{sen2008collective}.  However, not all ($E^t$,$\Gamma^t$) are correct in their predictions.  The false predictions may lead the update of embedding neural network to the wrong direction. Within the set of unlabeled graphs, different ($E^t$,$\Gamma^t$) could contribute differently to the update of embedding neural network. To this end, we make use of the idea of \emph{cautious iteration} \cite{mcdowell2007cautious} and \emph{curriculum learning} \cite{bengio2009curriculum}, and cautiously exploit a subset of ($E^t$,$\Gamma^t$) to update the parameters of IC in each iteration.  Specifically, in iteration $t$, we choose the $t\lambda$ most confident predicted labels in both IC and HC while ignoring the less confident predicted labels.  This operation continues until all the unlabeled samples have been utilized or until a predefined iteration number.  \eat{To further improve the efficiency, the parameters of IC are not re-trained but fine-tuned based on the parameters obtained in the previous iteration.}  This algorithm is called \underline{SE}mi-supervised gr\underline{A}ph c\underline{L}assification via \underline{C}autious \underline{I}teration (SEAL-CI) and is presented in Algorithm ~\ref{algo.naive}.  Note here $\mathcal{W}$ is the set of all the parameters of IC and HC. \textcolor{black}{In line $3$, SEAL is well trained to provide a foundation for SEAL-CI to find informative and accurate pseudo labels thereafter}. In line $6$, the training set for IC has been enlarged by $t\lambda$ instances and it is done by ``committing'' these instances' labels from their maximum probability.  In other words, the newly enrolled training instances are found by:
\begin{equation}
  T(t\lambda,(E,\Gamma)) = \textsf{top}(\max((E,\Gamma)),t\lambda),
\end{equation}
here function $\textsf{top}(\cdot,\lambda)$ is used to select the top $t\lambda$ instances and function $\max(\cdot)$ is used to select the maximum value in a set. 

\subsubsection{\textcolor{black}{Analysis of SEAL-CI}}\label{aoci}
\textcolor{black}{
In this part, we discuss about the convergence properties of proposed SEAL-CI algorithm. 
We first show the empirical risk is a monotonically decreasing function with respect to the iteration step.}
\begin{theorem}
    \emph{Let $\zeta(\mathcal{W})$ be the empirical risk, $\mathcal{W}^t(t=1,2,\cdots)$ be the parameter set of algorithm \ref{algo.naive} at iteration step $t$, $\zeta(\mathcal{W}^t)$ be the empirical risk at the corresponding iteration step. Then $\zeta(\mathcal{W}^t)$ is a monotonically decreasing function, namely,
    \begin{equation}
    \zeta(\mathcal{W}^{t+1}) \leq \zeta(\mathcal{W}^t),
    \end{equation}
}
    here $\zeta(\mathcal{W}^t)$ is short for $\zeta(G_l|\mathcal{W}^t)$ and the size of $G_l$ will be dynamically increased as we add pseudo-labeled instances. 
    \label{thm:mono}
\end{theorem}
The {\itshape cautious selection function $T$} in the pseudo labeling process provides a Bayesian prior for incorporating unlabeled instances, which can be formalized as follows:
\begin{equation}\label{pf}
\begin{aligned}
   \zeta(\mathcal{W}^{t+1}) = \frac{1}{L+1}(& \text{CE}(Y,E_l)+\text{CE}(Y,\Gamma_l) + \text{CE}(y'_t,e_t) \\
   & + \text{CE}(y'_t,\gamma_t)),
\end{aligned}
\end{equation}
here  $E_l = \{e_i\}_{i=1}^L$ and $\Gamma_l = \{\gamma_i\}_{i=1}^L$, $\text{CE}(Y,E_l) = \sum_{g_i\in G_l}\text{CE}(y_i,e_i)$, $\text{CE}(\cdot, \cdot)$ is the cross-entropy function. $y'_t$ denotes the pseudo label that we incorporate with $T(e_i,\gamma_i)$ and $\lambda =1$, i.e., we consider increasing the size of training set by one at each iteration. Let $T$ be short for $T(e_i,\gamma_i)$, we have
\begin{equation}\label{eta}
   \text{CE}(y_t,\gamma_t) = (L+U)\text{Cov}[\text{CE}(y'_t,\gamma_t),T] + \mathbb{E}[\text{CE}(Y,\Gamma_l)],
\end{equation}
\begin{equation}\label{ete}
\text{CE}(y_t,e_t) = (L+U)\text{Cov}[\text{CE}(y'_t,e_t),T] + \mathbb{E}[\text{CE}(Y,E_l)].
\end{equation}
Based on Eq.~\eqref{eta} and Eq.~\eqref{ete}, we can rewrite $\zeta(\mathcal{W}^{t+1})$:
\begin{equation}\label{Lp1}
\begin{aligned}
    \zeta(\mathcal{W}^{t+1}) &= \beta\{\text{Cov}[\text{CE}(y'_t,e_t),T]+\text{Cov}[\text{CE}(y'_t,\gamma_t),T]\}  \\
   &\;\;\;\;+\frac{1}{L}(\text{CE}(Y,E_l) +\text{CE}(Y,\Gamma_l))\\
   %&\;\;\;\;+ \text{CE}(y_t,\gamma_t)) \\
   &=\beta\{\text{Cov}[\text{CE}(y'_t,e_t),T]+\text{Cov}[\text{CE}(y'_t,\gamma_t),T]\}\\
   &\;\;\;\;+ \zeta(\mathcal{W}^t) \leq \zeta(\mathcal{W}^t),\\
   %&\leq \zeta(\mathcal{W}^t)
\end{aligned}
\end{equation}
here $\beta = \frac{L+U}{L+1}$. The last inequality holds as the labeling strategy $T$ is similar to a discrete sign function. It is easy to see the positive correlation between $\gamma_t$ and $T$. Besides, $\gamma_t$ and the pseudo labels $y'_t$ are positively correlated. As cross entropy term $\text{CE}(y'_t,\gamma_t)$ is negatively correlated to $\gamma_t$, the correlation between $T$ and $CE(y'_t,\gamma_t)$ is negative. Similarly, the correlation between $T$ and $CE(y'_t,e_t)$ is negative. We thus finish the proof of Theorem \ref{thm:mono}.

\textcolor{black}{As the empirical risk is a monotonically decreasing function and bounded (i.e., non-negative), we conclude that the proposed SEAL-CI algorithm has a convergence property.}

\section{EXPERIMENTS}\label{sec.exp}
We evaluate SEAL-CI method on synthetic, text and social network data sets.

\subsection{Synthetic Data}\label{bench.syn}

We evaluate the performance of SEAL-CI on synthetic data.  We first give a description of the synthetic generator, then visualize the learned embeddings.  Finally we compare our methods with baselines in terms of classification accuracy.

\subsubsection{Synthetic Data Generation}

The benchmark data set Cora \cite{mccallum2000automating} contains 2708 papers which are connected by the citation relationship.  We borrow the topological structure of Cora to provide the skeleton (i.e., edges) of our synthetic hierarchical graph.  Then we generate a set of graph instances, which serve as the nodes of this hierarchical graph.  Since there are 7 classes in Cora, we adopt 7 different graph generation algorithms, that is, Watts-Strogatz \cite{watts1998collective}, Tree graph, Erd{\H o}s-R{\'e}nyi \cite{erdos1960evolution}, Barbell \cite{herbster2007prediction}, Bipartite graph, Barab$\acute{a}$si-Albert graph \cite{bollobas2003mathematical} and Path graph, to generate 7 different types of graph instances, and connect them in the hierarchical graph.

Specifically, to generate a graph instance $g$, we randomly sample a number from $[100, 200]$ as its size $n$.  Then we generate its structure and assign the class label according to the graph generation algorithm.  In this step, the parameter $p$ in Watts-Strogatz, Erd{\H o}s-R{\'e}nyi, Bipartite graph and Barab$\acute{a}$si-Albert graph is randomly sampled from $[0.1, 0.5]$, the branching factor for Tree graph is randomly sampled from $[1, 3]$.  At last, to make this problem more challenging, we randomly remove $1\%$ to $20\%$ edges in the generated graph $g$.  The statistics of the generated graph instances are listed in Table \ref{tab:sgg}.

\begin{table}
  \caption{Statistics of generated graph instances}
  \label{tab:sgg}
  \scalebox{1.0}{
  \begin{tabular}{ccccc}
    \toprule
    \textbf{Type}&\textbf{Number}&\textbf{Nodes}&\textbf{Edges}&\textbf{Density} \\
    \midrule
	Watts-Strogatz&351&173&347&2.3\%\\
	Tree&217&127&120&1.5\%\\
	Erd{\H o}s-R{\'e}nyi&418&174&3045&20\%\\
	Barbell&818&169&2379&16.3\%\\
	Bipartite&426&144&1102&10.6\%\\
	Barab$\acute{a}$si-Albert&298&173&509&3.4\%\\
	Path&180&175&170&1.1\%\\
  \bottomrule
\end{tabular}
}
\end{table}

\subsubsection{Baselines and Metrics}\label{syn.base}
We use the following approaches as our baselines:
\begin{itemize}
\item GK-SVM \cite{shervashidze2009efficient}, which calculates the graphlet count kernel (GK) matrix, then GK-SVM feeds the kernel matrix into SVM \cite{hearst1998support}.

\item WL-SVM \cite{shervashidze2011weisfeiler}, which is similar as above but using the Weisfeiler-Lehman subtree kernel (WL).

\item graph2vec-GCN \cite{DBLP:journals/corr/NarayananCVCLJ17}, which embeds the graph instances by graph2vec and then feeds the embeddings to GCN.

% \item cautious-SAGE-Cheby, which is similar to SEAL-CI except that we replace GCN with Cheby-GCN \cite{defferrard2016convolutional}.

% \item active-SAGE-Cheby, which is similar to SEAL-AI except that we replace GCN with Cheby-GCN \cite{defferrard2016convolutional}.

\item SAGE \cite{jiawww19}, which ignores the connections between graph instances and treats them independently.

\item \textcolor{black}{GAT \cite{velivckovic2017graph}, which uses attention mechanisms in node-level neighborhood aggregation.}

\item \textcolor{black}{GIN \cite{xu2018powerful}, which is the GNN model with state-of-the art performance in graph classification.}

\item MIRACLE \cite{wang2021multi}, which uses the multi-view contrast learning method to exploit relation between the structure of graph instance level and the one at the hierarchical graph level.

\item SEAL is the base of SEAL-CI, which differs with SEAL-CI in two ways: 1) it does not consider label enlargement, and 2) it is trained in an end-to-end fashion with the given class labels.

\end{itemize}
We use 300 graph instances as the training set for all methods.  We use 1000 graph instances as the testing set.  We run each method 5 times and report its average accuracy.  The number of epochs for graph2vec is 1000 and the learning rate is 0.3. For SEAL-CI, we use SAGE \cite{jiawww19} as the graph pooling method.  We use a two-layer GCN in IC, in which the first GCN layer has 32 output channels and the second GCN layer has 4 output channels.  The dense layer has 48 units with a dropout rate of 0.3. 

\subsubsection{Results}
Table \ref{tab:synr} shows the experimental results for semi-supervised graph classification.  Among all approaches, SEAL-CI achieves the best performance.  In the following, we analyze the performance of all methods categorized into 3 groups.

\noindent\underline{Group *1}: All the embedding-based methods perform better than these two kernel methods, which proves that embedding vectors are effective representations for graph instances and are suitable input for graph neural networks. 

\noindent\underline{Group *2}: graph2vec-GCN achieves 85.2\% accuracy, which is comparable to that of SAGE, but lower than that of SEAL-CI.  One possible explanation is that graph2vec is an unsupervised embedding method, which fails to generate discriminative embeddings for classification.  Another possibility is that the 300 training instances do not include very informative ones. These limitations of graph2vec motivate us to use supervised graph representation modules such as SAGE and the label enlargement framework in SEAL-CI. MIRACLE and SEAL generate the graph instance representations in a supervised way, and they outperform graph2vec-GCN by more than 1.5\%.  

% \noindent\underline{Group *3}: cautious-SAGE-Cheby outperforms SAGE by only 0.8\%, which is not remarkable considering that it exploits many more training instances.  The accuracy of active-SAGE-Cheby is 3.3\% lower than that of SEAL-AI, which means that Cheby-GCN is inferior to GCN.

\noindent\underline{Group *3}: SEAL-CI outperforms MIRACLE and SEAL, which proves the effectiveness of our hierarchical graph perspective and the label enlargement algorithm for semi-supervised graph classification.  

\begin{table}
  \caption{Comparison of different methods on the synthetic data set for semi-supervised graph classification}
  \label{tab:synr}
  \centering
  \begin{tabular}{ccc}
    \toprule
     &\textbf{Algorithm}&\textbf{Accuracy} \\
    \midrule
		\multirow{3}{*}{*1}& \textbf{GK-SVM} & 77.8\%\\
		& \textbf{WL-SVM} & 83.4\%\\
		& \textbf{SAGE} & 85.7\% \\
		& \textcolor{black}{\textbf{GAT}} & 85.0\% \\
		& \textcolor{black}{\textbf{GIN}} & 86.7\% \\
		\hline
		\multirow{3}{*}{*2}& \textbf{graph2vec-GCN} &85.2\%\\
		& \textbf{MIRACLE} & 86.7\%\\
		& \textbf{SEAL} & 87.8\% \\
		\hline
		\multirow{1}{*}{*3} & \textbf{SEAL-CI} & 91.2\%\\
		%& \textbf{SEAL-AI} & \textbf{92.4}\%\\
	  \bottomrule
\end{tabular}
\end{table}

\subsection{Text Data}\label{textdata}
We evaluate SEAL-CI on the arXiv paper dataset. Firstly, We introduce the compositions of arXiv paper dataset and present how to construct hierarchical graph from text data. Then, we show the evaluation results and give some insights on how to construct a hierarchical graph for text data.

\subsubsection{Data Description}
arXiv is an open-access repository of electronic preprints. 
We collect 4666 Computer Science (CS) arXiv papers indexed by Microsoft Academic Graph (MAG) \cite{beltagy2019scibert}. These papers belong to five subject areas including AI, CL, IT, LG and CV. The arXiv papers data forms a citation network which indicates citation relationships. Each paper consists of two parts: title and abstract. The statistics of the arXiv paper dataset is listed in Table \ref{tab:s_arxiv}.

% ('arxiv cs ai', 'arxiv cs cl', 'arxiv cs it', 'arxiv cs lg', 'arxiv cs cv')
% There are about 16000 papers submitted per month as of April 2021. We have selected 169,343 Computer Science (CS) arXiv papers indexed by MAG \cite{} to construct this hierarchical dataset. These papers are belong to one of the 40 subject areas of arXiv CS papers, e.g., cs.AI, cs.LG, and cs.OS, which are manually determined (i.e., labeled) by the paper’s authors and arXiv moderators. And each paper consist of their title and abstract for extracting the feature.  

% We have 1,166,243 citing relationship that one paper cites another one and 1,266,243 co-author relationship which indicates the two papers have the same first author.

The connections among graph instances (i.e., papers) are provided by citation relations. Each paper is a semantic graph instance constructed from its title and abstract. 
There are three type of nodes in a semantic graph instance: abstract node, title node and argument nodes. 
Specifically, for each sentence in abstract, we parse it into a tuple with Semantic Role Labeling (SRL) toolkit developed by AllenNLP \footnote{https://demo.allennlp.org/semantic-role-labeling}.
For each tuple, we regard its elements as the argument nodes of the semantic graph. The node features of the semantic graph are extracted by pretrained Bert model \cite{devlin2019bert}.
The connections within the semantic graph are created as the following: 1) same tuple, if two nodes are from the same tuple, we connect them, and 2) vocabulary overlapping, we connect the node pair if the size of overlapped words is larger than half of the
minimum size of any text node.
In the Figure \ref{fig.arxiv_hier} we show an example of hierarchical graph of arXiv papers data.

\begin{figure*}
\begin{center}
\includegraphics [width=0.95\textwidth]{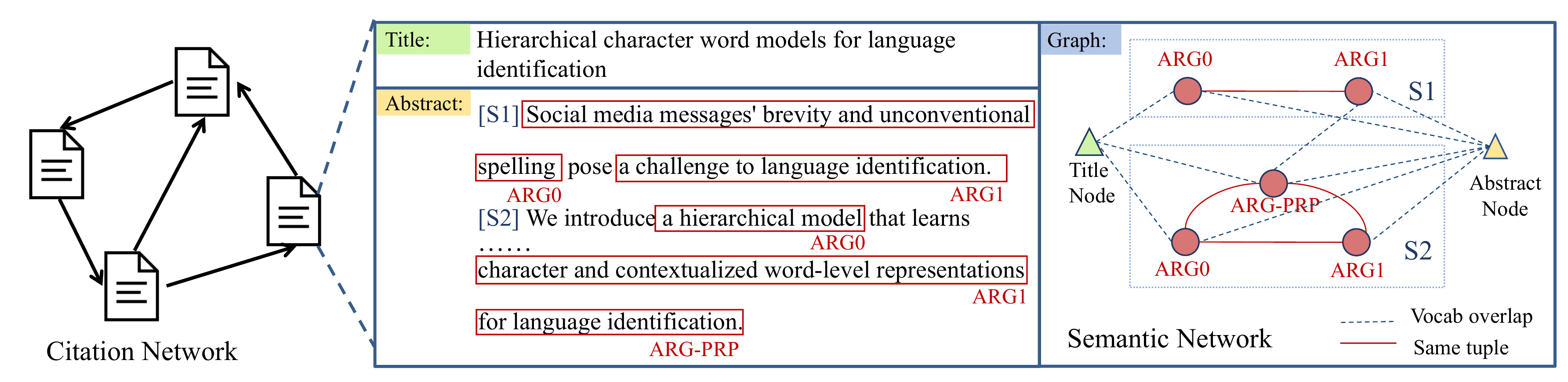}
\end{center}

\caption{An illustration of the hierarchical graph of arXiv papers data}
\label{fig.arxiv_hier}
\end{figure*}

%nodes feature construction

% cs.AI, cs.CL, cs.IT, cs.LG and cs.CV
% 10 30 28 24 16
% Counter({28: 909, 16: 1720, 24: 1157, 30: 648, 10: 232})
%label2category['arxiv category'][10], label2category['arxiv category'][30], label2category['arxiv category'][28], label2category['arxiv category'][24], label2category['arxiv category'][16] 
% 8.577586206896552
% 158.04310344827587
% 9.257716049382717
% 140.36265432098764
% 10.124312431243125
% 170.2948294829483
% 8.69057908383751
% 160.6266205704408
% 9.41046511627907
% 170.54302325581395
% print(sum(ai_10_tlen) / float(len(ai_10_tlen)))
% print(sum(ai_10_alen) / float(len(ai_10_alen)))
% print(sum(cl_30_tlen) / float(len(cl_30_tlen)))
% print(sum(cl_30_alen) / float(len(cl_30_tlen)))
% print(sum(it_28_tlen) / float(len(it_28_tlen)))
% print(sum(it_28_alen) / float(len(it_28_tlen)))
% print(sum(lg_24_tlen) / float(len(lg_24_tlen)))
% print(sum(lg_24_alen) / float(len(lg_24_tlen)))
% print(sum(cv_16_tlen) / float(len(cv_16_tlen)))
% print(sum(cv_16_alen) / float(len(cv_16_tlen)))

\begin{table}
  \caption{Statistics of arXiv paper data}
  \label{tab:s_arxiv}
  \centering
  \begin{tabular}{ccccc}
    \toprule
    \textbf{Class}&\textbf{Number}&\textbf{Length of title}&\textbf{Length of abstract}\\
    \midrule
	AI&232&8.57&158.04\\
	CL&648&9.25&140.36\\
	IT&909&10.12&170.29\\
	LG&1157&8.69&160.63\\
	CV&1720&9.41&170.54\\
  \bottomrule
\end{tabular}
\end{table}

\subsubsection{Baselines and Metrics}
In addition to  the set of baselines in Section \ref{syn.base}, we compare with the following Bert-feature baselines: 
\begin{itemize}
\item Bert-MLP, which classifies arXiv papers by feeding average features of title and abstract into a multi-layer perceptron (MLP).

\item Bert-IC, which feeds semantic graph instances into IC  and classifies arXiv papers in an independent graph classification fashion.

\item Bert-HC, which feeds average features of title and abstract into HC and classifies arXiv papers in a transductive node classification fashion.

\end{itemize}
400 graph instances are used as labeled training instances for all methods\eat{ except SEAL-AI, for which only 300 are used as labeled training instances at hand and then $B$ is set to 100 for active learning}.  We use 2000 instances for testing for all methods. We  set the dimension of Bert node feature to 768. We use average pool in IC. We set $\lambda=1$ and updated times $t=100$ for SEAL-CI\eat{, and $k=1$ for SEAL-AI}. We run each method 3 times and report its average accuracy.  

\begin{table}
  \caption{Comparison of different methods on arXiv Paper data for semi-supervised graph classification}
  \label{tab:arXiv}
  \centering
  \begin{tabular}{ccc}
    \toprule
     &\textbf{Algorithm}&\textbf{Accuracy} \\
    \midrule
		\multirow{3}{*}{*1}& \textbf{GK-SVM} & 38.7\%\\
		& \textbf{WL-SVM} & 35.1\%\\
		& \textbf{graph2vec-GCN} &42.7\%\\
		& \textcolor{black}{\textbf{GAT}} & 69.3\% \\
		& \textcolor{black}{\textbf{GIN}} & 69.1\% \\
		\hline
		\multirow{3}{*}{*2}& \textbf{Bert-MLP} & 63.7\%\\
		& \textbf{Bert-IC} & 70.1\%\\
		& \textbf{Bert-HC} & 75.2\%\\
		\hline
		\multirow{2}{*}{*3} 
		 & \textbf{MIRACLE} & 78.4\% \\
		 & \textbf{SEAL} & 78.8\% \\
% 		& \textbf{SAGE} &48.1\%\\
		\hline
		\multirow{1}{*}{*4}
		& \textbf{SEAL-CI} & 79.4\%\\
		%& \textbf{SEAL-AI} & \textbf{80.7}\%\\
	  \bottomrule
\end{tabular}
\end{table}

\subsubsection{Results}
Table \ref{tab:arXiv} shows the experimental results for semi-supervised
graph classification on arXiv paper data. SEAL-CI obtains the state-of-the-art performances among all baselines. To analyze the results of these baselines, all methods are divided into four group as the following. 

% And there are one group of methods based on Bert to discuss the influence of text feature extracted by Bert. 
\noindent\underline{Group *1}: Both GK-SVM and WL-SVM only utilize the structural information of semantic graph instances. And their classification accuracies are lower than 40\%, which indicates that structural information of semantic graph instances is not enough to obtain satisfied performance on text data. Graph2vec-GCN tries to exploit both structure of semantic graph instances and the hierarchical graph but still has a weak performance with an accuracy of 42.7\%.

\noindent\underline{Group *2}: In this group, we show the performance of several baselines based on Bert. Bert-MLP only uses Bert features extracted from title and abstract, and obtains an accuracy of 63.7\%. Bert-IC outperforms Bert-MLP by about 6.4\%, by combing both structural and Bert features of text data.
Bert-HC, by utilizing the connections between Bert features of title and abstract, achieves the best performance among baselines within Group *2.

\noindent\underline{Group *3}: Both SEAL and MIRACLE aim to minimize mutual information between graph instance representations and hierarchical representations and differ in the measurement of consistency. SEAL and MIRACLE outperform baselines in Group *1 with a margin of  30\% and Group *2 with a margin of more than 3\%. 

\noindent\underline{Group *4}: SEAL-CI outperforms MIRACLE and SEAL\eat{. Moreover, SEAL-AI outperforms SEAL-CI by about 1\%}, which demonstrates the effectiveness of the cautious iteration algorithm \ref{algo.naive}.

\subsubsection{Influence of the number of labeled training instances}
\textcolor{black}{We examine how the number of labeled training instances affects the performance of our methods.  We train Bert-IC, Bert-HC, SEAL and SEAL-CI with a label size of $\{100, 150, 200, 250, 300, 350, 400\}$.  \eat{We train SEAL-AI with 140 labeled instances and then set the budget $B$ for active learning at $\{0, 40, 80, 120, 160\}$.  Thus the three methods have the same number of labeled training instances. } We set $\lambda=1$ in SEAL-CI.  We run all methods 5 times, and plot their average accuracy in Figure \ref{fig.bud}.  As we can see from Figure \ref{fig.bud}, SEAL-CI performs the best since it can utilize more training samples.  \eat{As the number of labeled training instances increases, the performance of SEAL-AI improves dramatically.  SEAL-AI catches up with SEAL-CI at 260 labeled training instances and outperforms SEAL-CI at 300 labeled training instances.  It validates that SEAL-AI can make use of the iterations to find informative and accurate training samples.  Meanwhile SEAL-CI trusts the prediction of IC conditionally on its confidence, which may bring some noise to the learning process.}  SEAL outperform Bert-HC and Bert-IC in all cases, which makes sense because SEAL makes good use of the hierarchical graph setting and consider the connections between the graph instances for classification. Meanwhile, as SEAL-CI is based on SEAL, the good performance of SEAL provides a foundation for SEAL-CI to find informative and accurate pseudo labels, which could further boost the performance.  }

\begin{figure}
\begin{center}
\includegraphics [width=0.45\textwidth,scale=1]{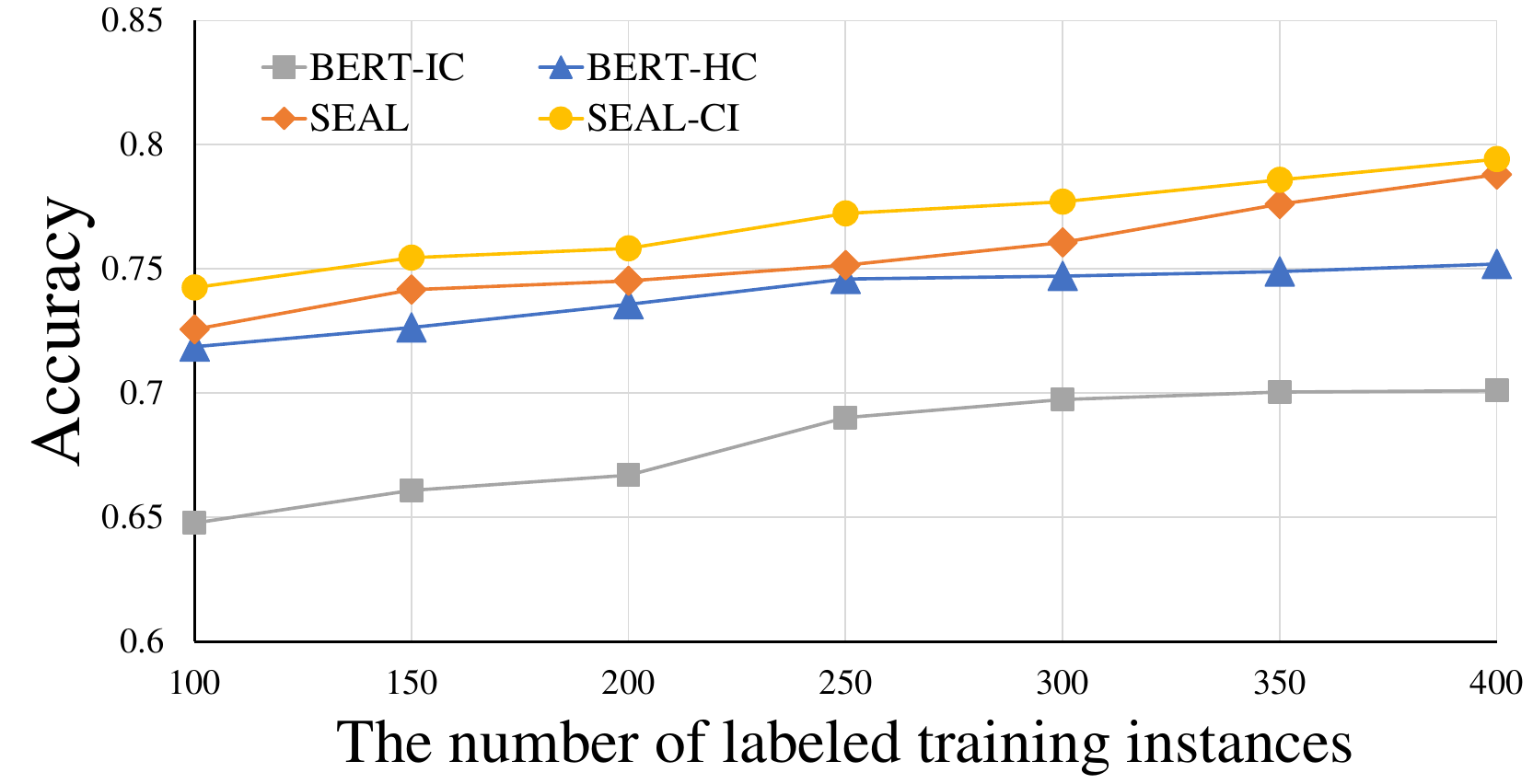}
\end{center}

\caption{Accuracy with different number of labeled training instances on text data for semi-supervised graph classification.}
\label{fig.bud}
\end{figure}

\subsubsection{Discussion}
How to construct hierarchical graphs from text data is an open question.  In the above experiment, we use semantic parsing techniques to construct graph instances for papers, and citation relationships to connect these papers. We think hierarchical graph structures could be a good paradigm to boost the performance of Natural Language Processing classification tasks. Moreover, it is worth mentioning that there are several other ways to construct semantic graphs such as Name Entity Recognition (NER).

\subsection{Social Network Data}
In this section, we evaluate SEAL-CI on Tencent QQ group data.  We describe the characteristics of this data set and then present the experimental results.  Finally, we have some open discussions on how to construct a hierarchical graph for social network data.

\subsubsection{Data Description}
Tencent QQ is a social networking platform in China with nearly 800 million monthly active users.  There are around 100 million active online QQ groups. In this experiment, we select 37,836 QQ groups with 18,422,331 unique anonymized users.  For each user, we extract seven personal features:
\begin{itemize}
\item number of days ever since the registration day;
\item most frequently active area code in the past 90 days;
\item number of friends;
\item number of active days in the past 30 days;
\item number of logging in the past 30 days;
\item number of messages sent in the past 30 days;
\item number of messages sent within QQ groups in the past 30 days.
\end{itemize}

We have 298,837,578 friend relationships among these users.  1,773 groups are labeled as ``game'' and the remaining groups are labeled as ``non-game''.

We construct the hierarchical graph from this Tencent QQ group data as follows.  A user is treated as an object, and a QQ group as a graph instance.  The users in one group are connected by their friendship.  The attribute matrix $X$ is filled with the attribute values of the users.  The statistics of the graph instances are listed in Table \ref{tab:sqq}.  We build the hierarchical graph from the graph instances via common members across groups, that is, if two groups have more than one common member, we connect them.

\begin{table}
  \caption{Statistics of collected Tencent QQ groups}
  \label{tab:sqq}
  \centering
  \begin{tabular}{ccccc}
    \toprule
    \textbf{Class label}&\textbf{Number}&\textbf{Nodes}&\textbf{Edges}&\textbf{Density} \\
    \midrule
	game&1,773&147&395&5.48\%\\
	non-game&36,063&365&1586&3.28\%\\
  \bottomrule
\end{tabular}
\end{table}

\subsubsection{Baselines and Metrics}
We use the same set of baselines as in Section \ref{syn.base}.  1000 graph instances are used as labeled training instances for all methods\eat{ except SEAL-AI, for which only 500 are used as labeled training instances at hand and then $B$ is set to 500 for active learning}.  We use 10,000 instances for testing for all methods.  We run each method 3 times and report its average performance.  For IC, we use SAGE and set the parameters as the same in Section \ref{bench.syn}. Since the class distribution is extremely imbalanced in this data set, we report the Macro-F1 instead of accuracy.

\subsubsection{Results}

Table \ref{tab:qq} shows the experimental results.  SEAL-C/AI outperform GK, WL and graph2vec by at least 36\% in Macro-F1.  
% Within our framework, GCN is better than Cheby-GCN for about 6\%. 
SEAL beats MIRACLE by about 4\%. SEAL-CI outperforms SEAL and MIRACLE with a margin of more than 1.2\%, which demonstrates the effectiveness of the cautious iteration algorithm \ref{algo.naive} on social network data.
\eat{SEAL-AI outperforms SEAL-CI by 1.6\%.  } Next we consider the relation between $\lambda$ and the false prediction rate (i.e., the percentage of misclassified instances) in the pseudo-labeling.  Figure \ref{fig.seal} shows the false prediction rate within the $\lambda$ most confident predictions of IC.  As we can see, the false prediction rate increases as $\lambda$ increases and it reaches $2.4\%$ when $\lambda=2000$. In the framework of SEAL-CI, as the iteration goes on, we shall bring in more noise to the parameter update, while all the training samples in SEAL are informative and correct. This observation tells us when using SEAL-CI, we need to be ``cautious'' enough, i.e., $\lambda$ should be small enough, to ensure the false prediction rate is small.

\eat{Next we provide the reason why SEAL-AI outperforms SEAL-CI on this data set.  Figure \ref{fig.seal} shows the false prediction rate (i.e., the percentage of misclassified instances) within the $\lambda$ most confident predictions of IC.  As we can see, the false prediction rate increases as $\lambda$ increases and it reaches $2.4\%$ when $\lambda=2000$. In the framework of SEAL-CI, as the iteration goes on, we shall bring in more noise to the parameter update of SAGE, while all the training samples in SEAL-AI are informative and correct. This explains why SEAL-AI outperforms SEAL-CI on this Tencent QQ group data.}

\begin{figure}
\begin{center}
\includegraphics [width=0.45\textwidth,scale=1]{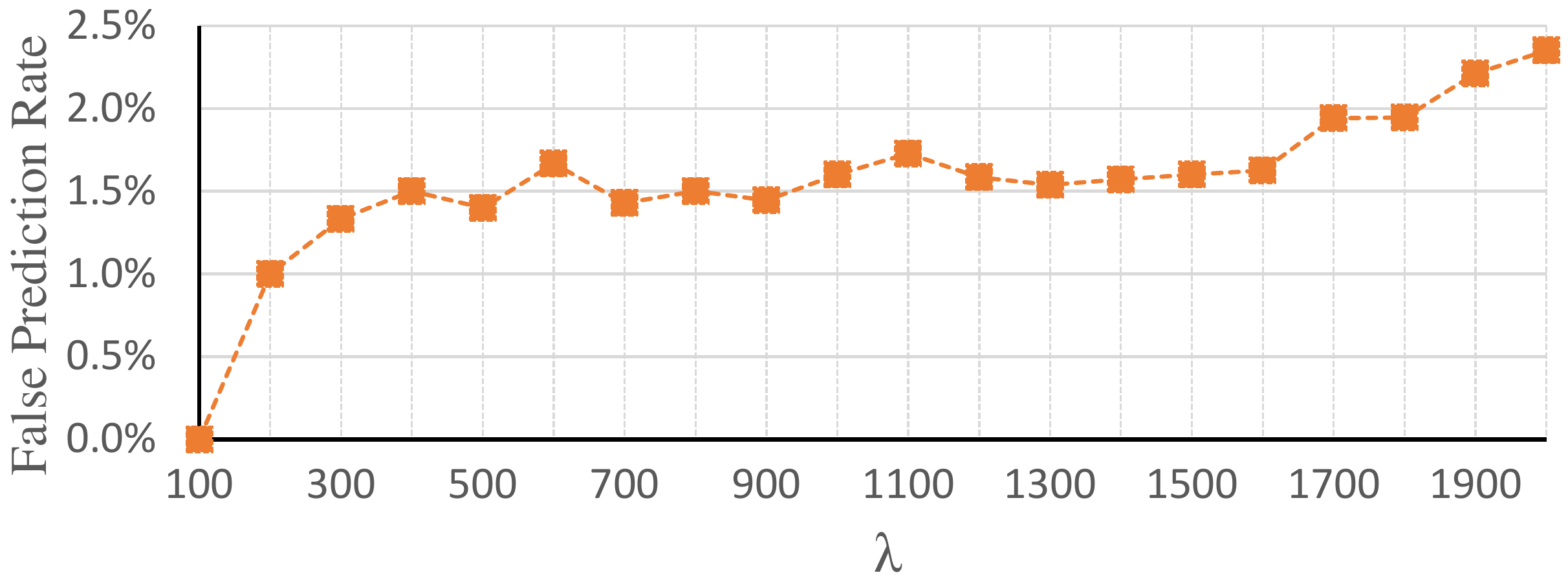}
\end{center}

\caption{The false prediction rate of IC with $\lambda$ in SEAL-CI.}
\label{fig.seal}
\end{figure}

\begin{table}
  \caption{Comparison of different methods on Tencent QQ group data for semi-supervised graph classification}
  \label{tab:qq}
  \centering
  \begin{tabular}{ccc}
    \toprule
     &\textbf{Algorithm}&\textbf{Macro-F1} \\
    \midrule
		\multirow{3}{*}{*1}& \textbf{GK-SVM} & 48.8\%\\
		& \textbf{WL-SVM} & 47.8\%\\
		& \textbf{SAGE} & 75.0\% \\
		& \textcolor{black}{\textbf{GAT}} & 79.9\% \\
		& \textcolor{black}{\textbf{GIN}} & 79.5\% \\
		\hline
		\multirow{3}{*}{*2}& \textbf{graph2vec-GCN} &48.1\%\\
		& \textbf{MIRACLE} & 80.3\%\\
		& \textbf{SEAL} & 84.3\%\\
		\hline
		\multirow{1}{*}{*3}  
		& \textbf{SEAL-CI} & 85.5\%\\
		%& \textbf{SEAL-AI} & \textbf{87.1}\%\\
	  \bottomrule
\end{tabular}
\end{table}

\subsubsection{Visualization}
We provide visualization of a ``game'' group and its neighborhood in Figure \ref{fig.gam}.  The left part is the ego network of the center ``game'' group.  In the one-hop neighborhood of this ``game'' group, there are 10 ``game'' groups and 19 ``non-game'' groups. ``Game'' groups are densely interconnected with a density of 34.5\%, whereas ``non-game'' groups are sparsely connected with a density of 8.8\%.  The much higher density among ``game'' groups validates that common membership is an effective way to relate them in a hierarchical graph for classification.  The right part depicts the internal structure of the ego ``game'' group with 22 users. A bigger node indicates a larger importance obtained by SAGE, and a darker green color implies a larger node degree.  These 22 members are loosely connected and there are no triangles.  This makes sense because in reality online ``game'' groups are not acquaintance networks.  Regarding the learned node importance, node 1 has the highest importance as it is the second active member and has a large degree in this group.  Node 16 is also important since it has the highest degree in this group.  The ``border'' member 5 has a big attention weight since it has the largest number of days ever since the registration day and is quite active in this group.
\begin{figure}
\begin{center}
\includegraphics [width=0.35\textwidth,scale=1]{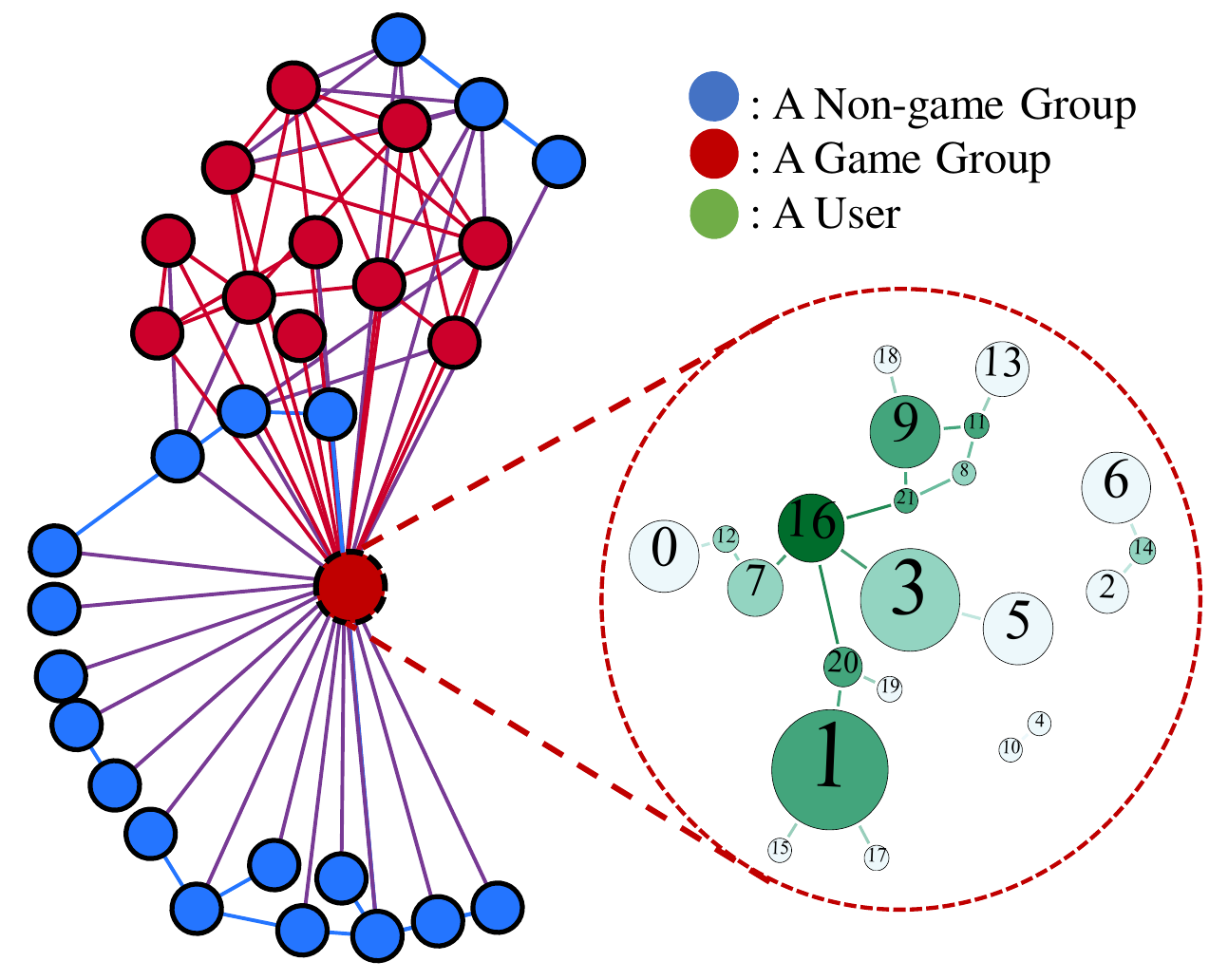}
\end{center}

\caption{The ego network of a ``game'' group.  The left side is the ego network, in which ``game'' groups are in red and ``non-game'' groups are in blue.  The right side is the internal structure of the ego ``game'' group, in which a bigger node indicates a larger importance, and a darker color implies a larger node degree.}
\label{fig.gam}
\end{figure}

\subsubsection{Discussion}
How to construct a hierarchical graph from social network is an open question.  In the above experiment, we connect two QQ groups if they have more than one common member (i.e., $>1$).  When we change the threshold, it directly affects the edge density in the hierarchical graph, and may influence the classification performance.  For example, if we connect two QQ groups when they have one common member or more (i.e., $\geq 1$), the edge density is 2.8\% compared with 0.27\% in the first setting.  A proper setting of this threshold is data dependent, and can be determined through a validation set.

\section{CONCLUSION}\label{sec.con}

In this paper, we study the problem of semi-supervised hierarchical graph classification. To enforce a consistency among different levels of the hierarchical graph,  we propose Hierarchical Graph Mutual Information (HGMI) and show HGMI can be computed via non-hierarchical graph mutual information computation methods.  We build two classifiers IC and HC at the graph instance level and the hierarchical graph level respectively to fully exploit the available information.  Our semi-supervised solution SEAL-CI adopts an iterative framework to update IC and HC alternately with an enlarged training set.  We present two hierarchical graph benchmarks and demonstrate that SEAL-CI outperforms other competitors by a significant margin.

\section{Acknowledgements}

The work described in this paper was supported by grants from HKUST(GZ) under a Startup Grant, HKUST-GZU Joint Research Collaboration Fund (Project No.: GZU22EG05) and Tencent AI Lab Rhino-Bird Focused Research Program RBFR2022008.

\ifCLASSOPTIONcaptionsoff
  \newpage
\fi



%\bibliographystyle{IEEEtran}

% \bibliographystyle{unsrt}
% \bibliography{sample}
\printbibliography %added

%\begin{IEEEbiography}[{\includegraphics[width=1in,height=1.25in,clip,keepaspectratio]{mshell}}]{Jia Li.jpg}

\begin{IEEEbiography}[{\includegraphics[width=1in,height=1.25in,clip,keepaspectratio]{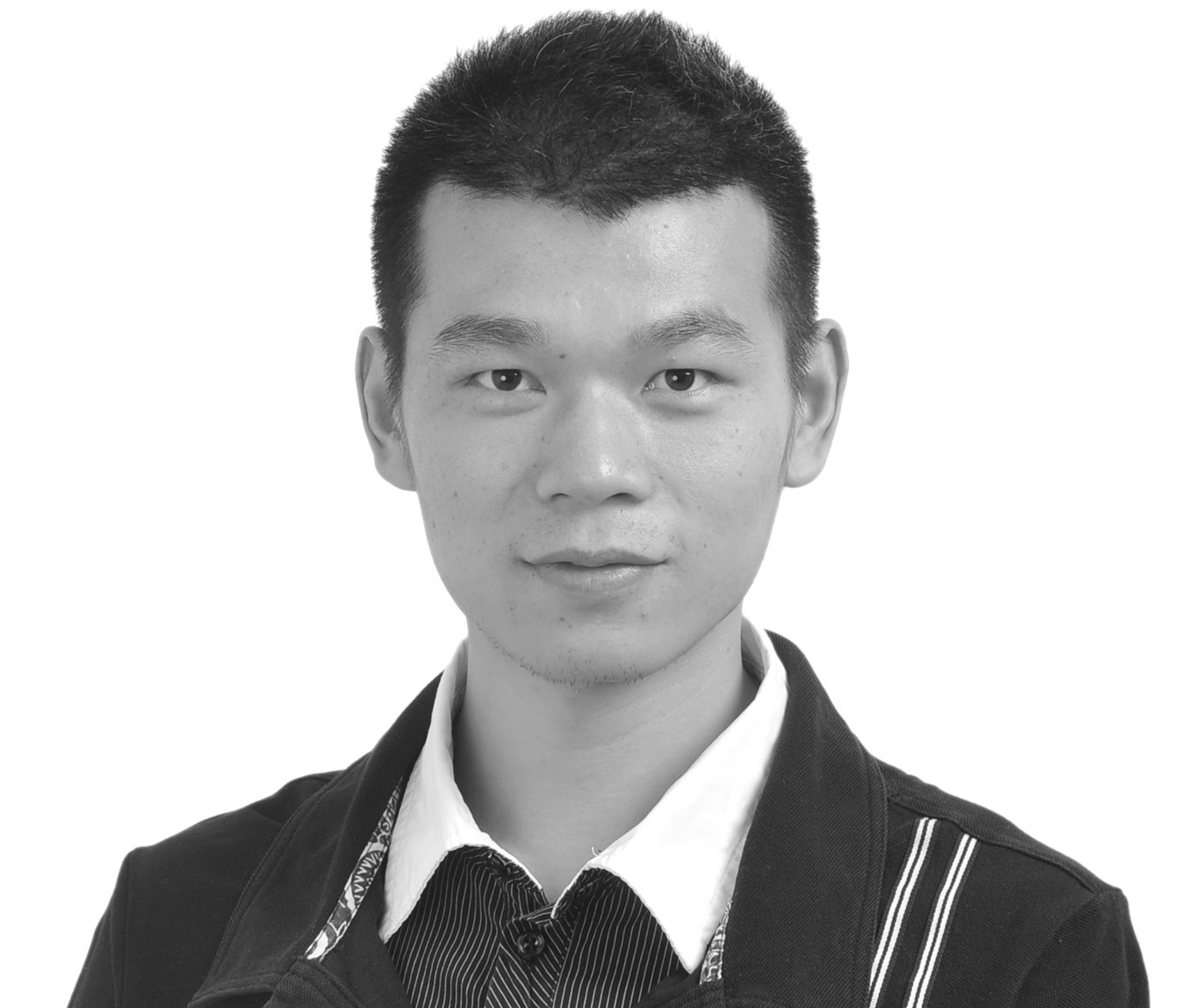}}]{Jia Li}
 is an assistant professor at Hong Kong University of Science and Technology (Guangzhou) and an affiliated assistant professor at Hong Kong University of Science and Technology. He received the Ph.D. degree at Chinese University of Hong Kong in 2021. His research interests
include machine learning, data mining and deep graph learning. He has published several papers as the leading author in top conferences such as KDD, WWW and NeurIPS.
\vspace{-10mm}
\end{IEEEbiography}

% \begin{IEEEbiographynophoto}[{\includegraphics[width=1in,height=1.25in,clip,keepaspectratio]{fig/hardenhuang.png}}]{Yongfeng Huang} 
% is preparing to pursue  Ph.D. Degree in the The Hong Kong University of Science and Technology. Previously, he received master degree in Tsinghua University. His research interests include natural language processing and deep graph learning.
\begin{IEEEbiography}[{\includegraphics[width=1in,height=1.25in,clip,keepaspectratio]{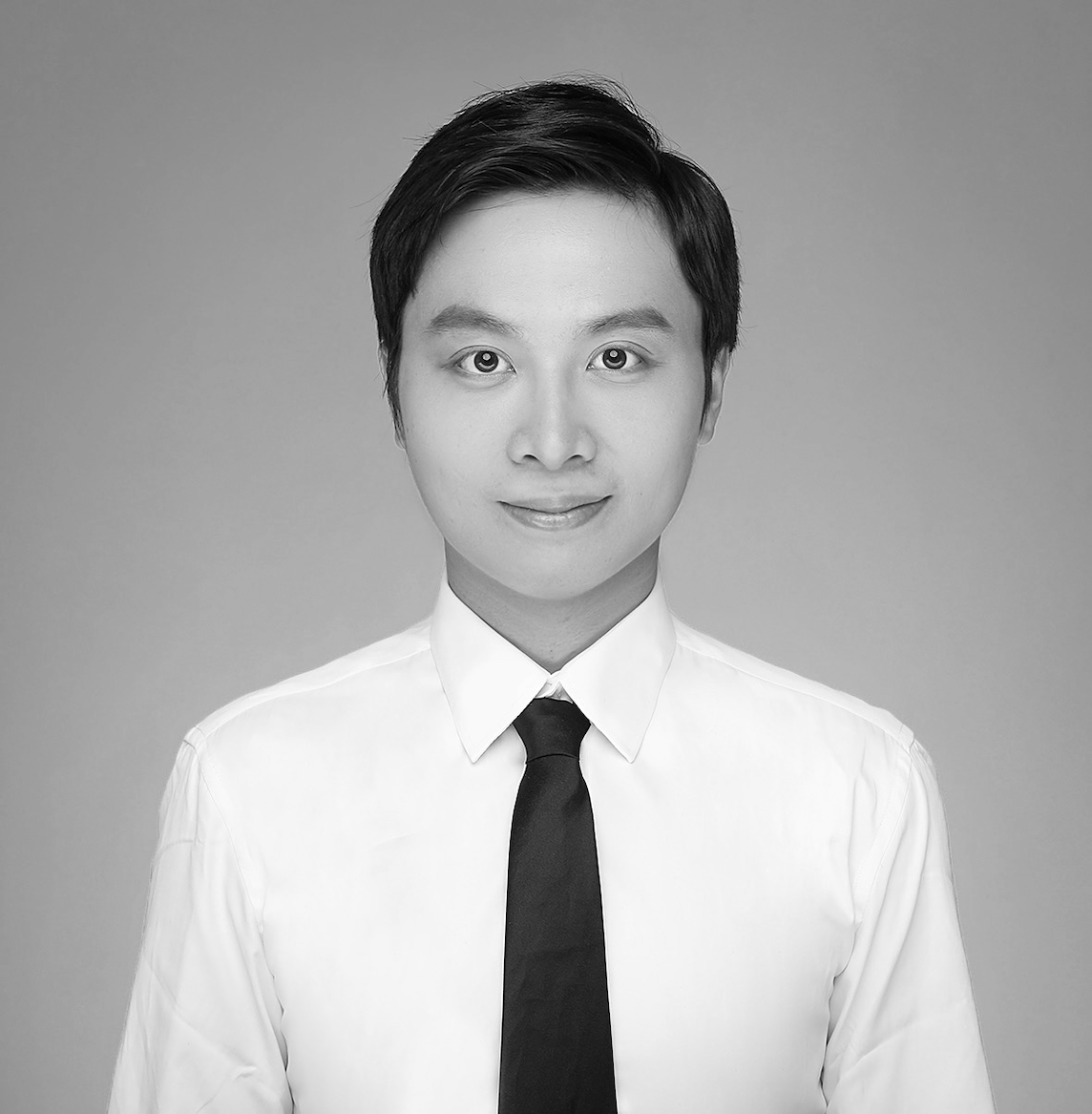}}]{Yongfeng Huang}
is a research assistant at The Hong Kong University of Science and Technology (Guangzhou). Previously, he received master degree in Tsinghua University. His research interests include natural language processing and deep graph learning.
\vspace{-10mm}
\end{IEEEbiography}

%\newpage

\begin{IEEEbiography}
% \vspace{-10mm}
[{\includegraphics[width=1in,height=1.25in,clip,keepaspectratio]{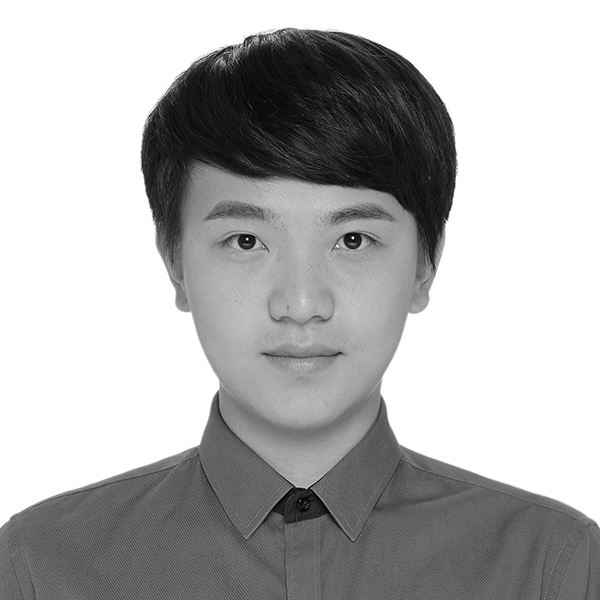}}]{Heng Chang}
is currently pursuing a Ph.D. Degree at Tsinghua University. He received his B.S. from the Department of Electronic Engineering, Tsinghua University in 2017. His research interests focus on graph representation learning and machine learning on graph data, especially in the development of efficient graph learning models, and the 
explanation and robustness of graph neural networks. He has published several papers in prestigious conferences including NeurIPS, AAAI, CIKM, etc.
\vspace{-10mm}
\end{IEEEbiography}



\begin{IEEEbiography}[{\includegraphics[width=1in,height=1.25in,clip,keepaspectratio]{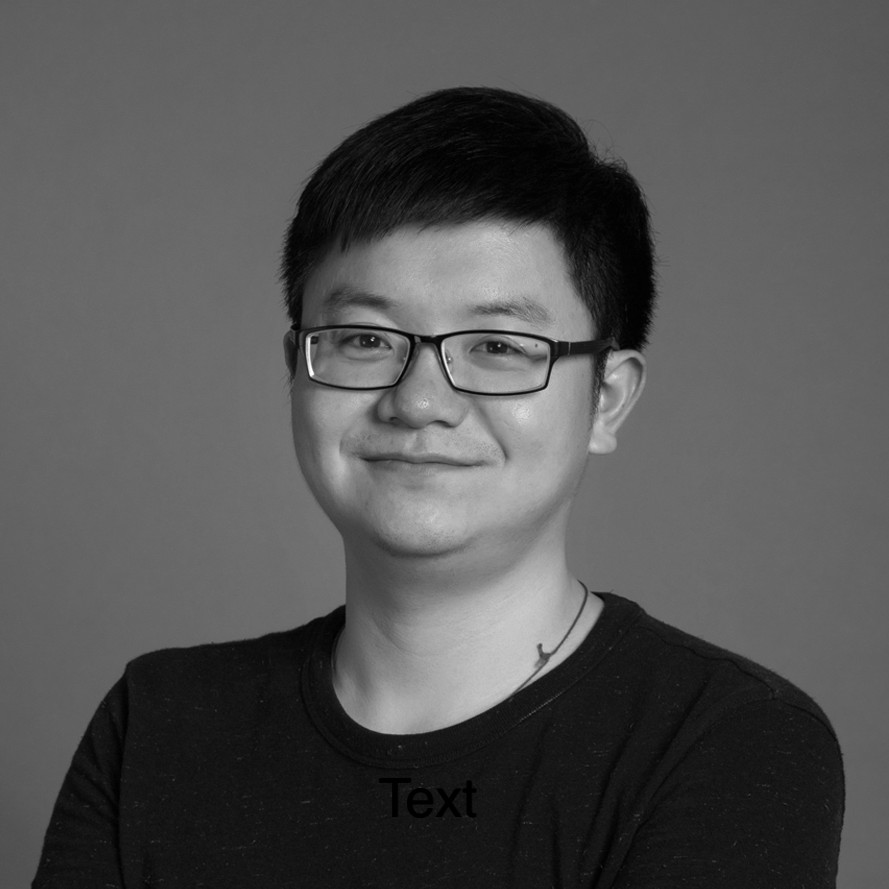}}]{Yu Rong} is a senior researcher of Machine Learning Center in Tencent AI Lab. He received the Ph.D. degree from The Chinese University of Hong Kong in 2016. He joined Tencent AI Lab in June 2017.  His main research interests include social network analysis, graph neural networks, and large-scale graph systems. In Tencent AI Lab, he is working on building the large-scale graph learning framework and applying the deep graph learning model to various applications, such as ADMET prediction and malicious detection. He has published several papers on data mining, machine learning top conferences, including KDD, WWW, NeurIPS, ICLR, CVPR, ICCV.
\vspace{-10mm}
\end{IEEEbiography}

%\begin{IEEEbiography}[{\includegraphics[width=1in,height=1.25in,clip,keepaspectratio]{fig/junzhou.pdf}}]{Junzhou Huang}
% is an Associate Professor in the Computer Science and Engineering department at the University of Texas at Arlington. He also served as the director of machine learning center in Tencent AI Lab. He received the B.E. degree from Huazhong University of Science and Technology, Wuhan, China, the M.S. degree from the Institute of Automation, Chinese Academy of Sciences, Beijing, China, and the Ph.D. degree in Computer Science at Rutgers, The State University of New Jersey. His major research interests include machine learning, computer vision and imaging informatics. He was selected as one of the 10 emerging leaders in multimedia and signal processing by the IBM T.J. Watson Research Center in 2010. His work won the MICCAI Young Scientist Award 2010, the FIMH Best Paper Award 2011, the MICCAI Young Scientist Award Finalist 2011, the STMI Best Paper Award 2012, the NIPS Best Reviewer Award 2013, the MICCAI Best Student Paper Award Finalist 2014 and the MICCAI Best Student Paper Award 2015. He received the NSF CAREER Award in 2016. 
%\end{IEEEbiography}

%\vfill


\end{document}